\documentclass{article}

\usepackage{amsmath}
\usepackage{amssymb}
\usepackage{amsthm}
\usepackage{alltt}
\usepackage{eufrak}
\usepackage{stmaryrd}
\usepackage{nicefrac}

\pdfmapline{=xbbold <./xbbold.pfb}
\pdfmapline{=upml <./upml.pfb}

\DeclareFontFamily{U}{upml}{}
\DeclareFontShape{U}{upml}{m}{n}{<-6>upml<6-8>upml<8->upml}{}
\DeclareSymbolFont{UPML}{U}{upml}{m}{n}

\DeclareMathSymbol{\blwedge}{\mathrel}{UPML}{"63}
\DeclareMathSymbol{\blvee}{\mathrel}{UPML}{"64}
\DeclareMathSymbol{\blRightarrow}{\mathrel}{UPML}{"69}
\DeclareMathSymbol{\blLeftrightarrow}{\mathrel}{UPML}{"65}
\DeclareMathSymbol{\blneg}{\mathord}{UPML}{"6E}
\DeclareMathSymbol{\blotimes}{\mathbin}{UPML}{"6F}

\DeclareFontFamily{U}{xbbold}{}
\DeclareFontShape{U}{xbbold}{m}{n}{<-6>xbbold<6-8>xbbold<8->xbbold}{}
\DeclareMathAlphabet{\blmathbb}{U}{xbbold}{m}{n}

\let\mathbb=\blmathbb

\newcommand{\DIN}[1][\relax]{%
  \ensuremath{%
    \ifthenelse{\equal{#1}{\relax}}{%
      \boldsymbol{\mathcal{D}}}{%
      \boldsymbol{\mathcal{D}}_{#1}}}}

\newcommand{\Tupl}[1]{\ensuremath{\mathop{\mathrm{Tupl}}(#1)}}

\newcommand{\RAeadom}[1]{\ensuremath{\mathop{\mathcal{E}_{#1}}}}
\newcommand{\RAeadomEV}[1]{\ensuremath{\mathop{\mathcal{E}^{\DIN}_{#1}}}}
\newcommand{\Dee}[1]{\ensuremath{#1_\emptyset}}

\newcommand{\join}[3]{\ensuremath{#2\bowtie_{#1}#3}}
\newcommand{\project}[2]{\ensuremath{\pi_{#1}(#2)}}
\newcommand{\frdiv}{\div}
\newcommand{\rdiv}[3]{\ensuremath{#2 \frdiv^{#1} #3}}

\newcommand{\rename}[2][f]{\ensuremath{\rho_{\bgroup\def\as##1##2{##2 \leftarrow ##1}#1\egroup}(#2)}}
\newcommand{\TupCommon}[3]{\ensuremath{r_{#1 #2}^{\between #3}}}
\def\itm#1{{\rm(\textit{\romannumeral#1})}}
\def\nltimes{\mathrel{\bar{\ltimes}}}

\def\circop{\mathop{\circ}}
\newcommand{\tv}[1]{\ensuremath{\blmathbb{#1}}}

\newcommand{\ptcsup}[0]{
  \bigvee_{\!\mathbf{r}_1} \mathcal{T}(\mathbf{r}_1, \mathbf{r}_2)
}
\newcommand{\ptcinf}[0]{
  \bigwedge_{\mathbf{r}_1} \mathcal{T}(\mathbf{r}_1, \mathbf{r}_2)
}

\newcommand{\ptcsupprime}[0]{
  \bigvee_{\!\mathbf{r}_1} \mathcal{T'}(\mathbf{r}_1, \mathbf{r}_2)
}
\newcommand{\ptcinfprime}[0]{
  \bigwedge_{\mathbf{r}_1} \mathcal{T'}(\mathbf{r}_1, \mathbf{r}_2)
}

\newcounter{global}
\theoremstyle{definition}
\newtheorem{definition}[global]{Definition}
\theoremstyle{plain}
\newtheorem{theorem}[global]{Theorem}
\newtheorem{lemma}[global]{Lemma}

\newtheorem{corollary}[global]{Corollary}
\newtheoremstyle{note}{}{}{}{}{\itshape}{.}{.5em}{}
\theoremstyle{note}
\newtheorem{remark}{Remark}%
\renewcommand\section{%
  \@startsection {section}{1}{\z@}%
  {-3.5ex \@plus -1ex \@minus -.2ex}%
  {2.3ex \@plus.2ex}%
  {\normalfont\large\bfseries}}
\begin{document}

\title{Relational Division in Rank-Aware Databases}

\author{Ondrej Vaverka\footnote{%
    e-mail: \texttt{ondrej.vaverka@upol.cz},
    phone: +420 585 634 705,
    fax: +420 585 411 643} \and Vilem Vychodil}

\date{\normalsize%
  Dept. Computer Science, Palacky University Olomouc}

\maketitle

\begin{abstract}
  We present a survey of existing approaches to relational division
  in rank-aware databases, discuss issues of the present approaches,
  and outline generalizations of several types of classic division-like
  operations. We work in a model which generalizes the Codd model of data
  by considering tuples in relations annotated by ranks, indicating degrees
  to which tuples in relations match queries. The approach utilizes complete
  residuated lattices as the basic structures of degrees.
  We argue that unlike the classic model, relational divisions are
  fundamental operations which cannot in general be expressed by means of
  other operations. In addition, we compare the existing and proposed operations
  and identify those which are faithful counterparts of universally
  quantified queries formulated in relational calculi. We introduce
  Pseudo Tuple Calculus in the ranked model which is further used to
  show mutual definability of the various forms of divisions
  presented in the paper.

  % \keywords{%
  %   database,
  %   division,
  %   lattice,
  %   rank,
  %   relation,
  %   residuum}
\end{abstract}

%%%%%%%%%%%%%%%%%%%%%%%%%%%%%%%%%%%%%%%%%%%%%%%%%%%%%%%%%%%%%%%%%%%%%%%%%%%%%%%%
%%%%%   INTRODUCTION
%%%%%%%%%%%%%%%%%%%%%%%%%%%%%%%%%%%%%%%%%%%%%%%%%%%%%%%%%%%%%%%%%%%%%%%%%%%%%%%%
\section{Introduction}
In this paper, we present a survey and new results in the area of division-like
operations in rank-aware relational models of data. In particular, we are interested
in models which allow \emph{imperfect matches of queries} in addition to the usual
precise yes/no matches of queries. By an ``imperfect match'' we mean a situation
where given record in a database does not match a query in the usual sense but the
record is sufficiently close to a (hypothetical) record that matches the
query exactly.
In many situations, it is desirable to include records with imperfect matches in
the result of a query and introduce \emph{scores}
which indicate the degrees to which the
records match the given query. For instance, in a database of products, we may
query for products with price equal to $\$1,200$. In the traditional understanding,
a product sold for $\$1,198$ does not match the query. Nevertheless, we may want
to include such product in the result and annotate it with a high score indicating
that the product matches the query ``almost perfectly'' but not fully. In fact,
reasoning with imperfect matches is inherent to human thinking and human
perception of concepts like the proximity of values.
Rank-aware databases~\cite{Il:ASOTQPTiRDS,LiChCh:RQAaOfRTQ} and
related models of data aim at such reasoning with imperfect matches and are concerned
with its formalisation, analysis, and implementation in computer database systems.

Our investigation of division-like operations is motivated by the fact that in
most of the existing rank-aware approaches to databases, discussion of such
operations is either completely omitted or focuses only on particular Codd-style
divisions. Indeed, compared to
operations like projections and joins, the current rank-aware approaches pay
little or no attention to division-like operations.

There seem to be two reasons for the absence of discussions of divisions in
rank-aware models:
First, a proposed rank-aware model simply omits divisions because its authors
do not consider such an operation important. Second, the authors of a rank-aware
model expect a division-like operation to be definable by the remaining operations
in a similar way as in the classic relational model of data. We argue that neither
of the points is tenable and divison-like operations deserve our attention:

%%%%%%%%%%%%%%%%%%%%%%%%%%%%%%%%%%%%%%%%%%%%%%%%%%%%%%%%%%%%%%%%%%%%%%%%%%%%%%%%
\paragraph{1) Divisions are important}
Division-like operations are considered in relational query systems
in order to express queries which take form of particular
categorical propositions.
It is well understood that classic relational queries of the form
``some $\varphi$ is $\psi$'' can be expressed by means of combinations
of projections and natural joins which are known as semijoins. Analogous
queries can also be considered in rank-aware approaches with the same
meaning except for the fact that the results of queries are annotated by
scores.
Naturally, one should expect to be able to formulate queries of the form
of categorical proposition ``all $\varphi$ are $\psi$'' in a rank-aware model.
In the classic model, such queries are expressed by division-like operations.
In addition, some variants of the classic relational division have a close
relationship to the notions of containment (subsethood) of relations.
From this viewpoint, one should expect that containments and divisions in
a rank-aware model should both be defined and related as in the classic model.
Note that division-like operations are also interesting from the
data-analytical point of view.
For instance, concept-forming operators in formal concept
analysis~\cite{GaWi:FCA} can be seen as particular relational divisions.

%%%%%%%%%%%%%%%%%%%%%%%%%%%%%%%%%%%%%%%%%%%%%%%%%%%%%%%%%%%%%%%%%%%%%%%%%%%%%%%%
\paragraph{2) Divisions in rank-aware models are fundamental operations}
If a rank-aware model contains operations of difference (relational minus),
projection, and natural join, one may argue that a Codd-style
division~\cite{Co:Armodflsdb} is a definable operation in the ranked model
in much the same way as it is definable in the classic model. While in the
classic model, reasonable division-like operations can indeed be derived,
we show further in the paper that this assumption cannot be universally adopted
in rank-aware models. Technically, the operation can be defined as in the
ordinary case but in many cases it lacks the basic properties of
``reasonable division'' and no longer is a faithful representation of
queries of the form of categorical propositions ``all $\varphi$ are $\psi$''.
As a matter of fact, we argue in the paper that suitable variants of
divisions (or equivalent formalisms) should be included as fundamental
operations in rank-aware models.

\medskip
In this paper we focus on divisions from the perspective of a relational
model which can be seen as a generalization of the Codd~\cite{Co:Armodflsdb}
model of data from the point of view of residuated structures of degrees.
The basic idea of the model is that tuples in relations are annotated
by scores indicating degrees to which tuples match queries
analogously as in~\cite{Fa98:CFIfMS,Fagin:Oaafm},
cf. also \cite{LiChCh:RQAaOfRTQ}
introducing RankSQL and a survey paper~\cite{Il:ASOTQPTiRDS}. Our model
differs in how we approach the structures of scores and, consequently,
the underlying logic of imperfect matches. We use structures of degrees
which are recognized by fuzzy logics in the
\emph{narrow sense}~\cite{CiHaNo1,CiHaNo2,GaJiKoOn:RL,Got:Mfl,Haj:MFL} and
the principle of \emph{truth functionality} because our intention
is to develop the model so that particular issues handled in the model
(like querying and data dependencies) can be analyzed in terms of
logical deduction in the narrow sense. This is in contrast with
various approaches that appeared
earlier~\cite{Bo:Opdfr,BoPiRo:Crdfr,BoDuPiPr:Fqrdedo,DuPr:Sqofrd}
and utilized techniques from fuzzy sets (in the wide sense) where
the connection to residuated structures of degrees is not so strict.
We argue in the paper that the role of residuated structures is crucial
for a sound treatment of division-like operations.

Our paper is organized as follows.
In Section~\ref{sec:prelim}, we recall basic notions of our model.
In Section~\ref{sec:appr}, we survey existing and propose new
approaches to division operations in the classic as well as in the
graded setting. In Section~\ref{sec:ptc}, we introduce a query language
called Pseudo Tuple Calculus (PTC) that enables us to reason about the
operations with ease. Finally, in Section~\ref{sec:more}, we utilize PTC
to derive further observations on the mutual definability of the division
operations described in the paper.

%%%%%%%%%%%%%%%%%%%%%%%%%%%%%%%%%%%%%%%%%%%%%%%%%%%%%%%%%%%%%%%%%%%%%%%%%%%%%%%%
%%%%%   RELATIONAL MODEL BASED ON RESIDUATED STRUCTURES
%%%%%%%%%%%%%%%%%%%%%%%%%%%%%%%%%%%%%%%%%%%%%%%%%%%%%%%%%%%%%%%%%%%%%%%%%%%%%%%%
\section{Relational Model Based on Residuated Structures}\label{sec:prelim}
In this section, we present a survey of utilized notions from residuated
structures of degrees and fuzzy relational systems. Furthermore, we introduce
the basic notions of the generalized relational model of data and
its relational algebra~\cite{BeVy:TODS}.

%%%%%%%%%%%%%%%%%%%%%%%%%%%%%%%%%%%%%%%%%%%%%%%%%%%%%%%%%%%%%%%%%%%%%%%%%%%%%%%%
\subsection{Structures of Degrees}
We use complete residuated lattices as structures of degrees which represent
scores assigned to tuples and indicating degrees to which tuples match queries.
A \emph{residuated lattice}~\cite{Bel:FRS,GaJiKoOn:RL,Haj:MFL} is
a general algebra~\cite{Wec:UAfCS} of the form
\begin{align}
  \mathbf{L}=\langle L,\wedge,\vee,\otimes,\rightarrow,0,1\rangle
  \label{eqn:reslat}
\end{align}
such that $\langle L,\wedge,\vee,0,1 \rangle$ is
a bounded lattice~\cite{Bir:LT} with $0$ and $1$ being the least and
the greatest element of $L$, respectively;
$\langle L,\otimes,1 \rangle$ is a commutative monoid
(i.e., $\otimes$  is commutative, associative,
and $a \otimes 1 = 1 \otimes a = a$ for each $a \in L$);
$\otimes$ (a multiplication) and $\rightarrow$ (a residuum)
satisfy the \emph{adjointness property}:
\begin{align}
  a \otimes b \leq c \text{ if{}f } a \leq b \rightarrow c
  \label{eqn:adj}
\end{align}
for each $a,b,c \in L$ where $\leq$ is the order induced by the
lattice structure of $\mathbf{L}$ (i.e., $a \leq b$ if{}f $a = a \wedge b$).
A residuated lattice~\eqref{eqn:reslat} is called \emph{complete}
if its lattice part is a complete lattice,
i.e., if $L$ contains infima (greatest lower bounds) and suprema
(least upper bounds) of arbitrary subsets of $L$.
The multiplication $\otimes$ and its adjoint residuum $\rightarrow$
can be seen as general aggregation functions which interpret general
``conjunction'' and ``implication'' of scores, respectively. That is,
if a tuple matches query $Q_1$ with a score $a_1$ and
it also matches query $Q_2$ with a score $a_2$, then
$a_1 \otimes a_2$ may be interpreted as the score to which the tuple
matches the composed \emph{conjunctive query}
``$Q_1 \mathop{\text{\emph{and}}} Q_2$.''
This way the aggregation function is understood in~\cite{Fa98:CFIfMS}.
In a similar way, $a_1 \rightarrow a_2$ may be interpreted as the score
to which the tuple matches the composed \emph{conditional query}
``$\mathop{\text{\emph{if}}} Q_1 \mathop{\text{\emph{then}}} Q_2$.''

A typical choice of a complete residuated lattice $\mathbf{L}$ is a structure
given by a left-continuous triangular norm~\cite{KMP:TN}.
That is, $L=[0,1]$ (real unit interval), $\wedge$ and $\vee$ are minimum and
maximum (in which case the induced $\leq$ is the genuine ordering of reals),
and $\otimes$ is a left-continuous triangular norm.
The left-continuity of $\otimes$ ensures there is a residuum $\rightarrow$
satisfying \eqref{eqn:adj} which is in addition uniquely given by
\begin{align}
  a \rightarrow b =
  \textstyle\bigvee\bigl\{c \in L \,|\, a \otimes c \leq b\bigr\}.
  \label{eqn:resd_from_tnorm}
\end{align}
In words, \eqref{eqn:resd_from_tnorm} says that $a \rightarrow b$ is the
supremum of all $c \in L$ such that $a \otimes c \leq b$ (it can be shown
that $a \rightarrow b$ is in fact the greatest $c \in L$ satisfying such
property).

From pragmatic standpoints, the most important complete residuated lattices
are exactly those on the real unit interval given by \emph{continuous}
triangular norms.
All such structures can be obtained by constructing ordinal
sums~\cite{Bel:FRS,Haj:MFL,KMP:TN} of (isomorphic copies of)
three basic pairs of multiplications (and their corresponding residua):
$a \otimes b = \max(a + b - 1, 0)$ (\L ukasiewicz multiplication),
$a \otimes b = \min(a,b)$ (G\"odel or minimum multiplication),
$a \otimes b = a \cdot b$ (Goguen or product multiplication).

\begin{remark}
  The role of residuated lattices as general structure of truth degrees in
  truth-functional logics has been recognized by Goguen~\cite{Gog:Lic}.
  Important logics based on subclasses of residuated lattices include
  H\"ohle's monoidal logic~\cite{Ho:ML}, Basic Logic~\cite{Haj:MFL},
  and Monoidal T-norm Logic~\cite{EsGo:MTL}.
  Note that the truth-functionality is a crucial property which is
  not present in other models which also involve ranks like the probabilistic
  extensions of the Codd model, see~\cite{DaReSu:Pdditd}.
  In fact, the probabilistic databases tackle completely different issues and
  deal with \emph{uncertain data} which is \emph{not} our case because
  the approaches we discuss here deal with certain data and imperfect
  matches of queries.
\end{remark}

An important aspect of the relational model which is relevant to our paper
is that the classic relational model is based on the classic predicate
logic~\cite{Dat:DRM}. As a result, finite relations (informally
represented by ``data tables'') are used to represent both the base data
and results of queries. In fact, database instances (i.e., collections
of relations interpreting relational symbols/variables) can be seen as
predicate structures~\cite{Me87}, predicate formulas can be seen as queries,
and their interpretation in database instances corresponds to query evaluation.
Thus, the structures of truth values of the classical
predicate logic---the Boolean
algebras, are vital for the model and, loosely speaking, determine laws that
hold in the relational model.

The model we use in this paper can be seen as a relational
model of data which results from the classic one by replacing the Boolean
algebras with complete residuated lattices. This change has, of course,
its implications. First, we shift from structures with only yes/no matches
to structures which allow us to work with general (intermediate)
degrees---this is a desirable property for development of a rank-aware model.
Second, some laws that hold in the classic model are no longer valid
(e.g., \emph{tertium non datur}).
The second point shall be understood as virtue of the model rather than
a vice---note that Basic Logic extended by \emph{tertium non datur}
collapses into the classical logic~\cite{Haj:MFL}. In fact,
there are no proper fuzzy logics which satisfy
\emph{tertium non datur}. Our rationale for using (complete)
residuated lattices as the structures of degrees is that they represent
more general structures than the Boolean algebras which allow us to deal with
intermediate degrees and are still reasonably strong.

\begin{remark}
  Let us note that logics based on residuated lattices are used to reason about
  general scores. If $0$ and $1$ are used as the only scores, the logic collapses
  into the classic Boolean logic which is a desirable property. Also, the
  structures and operations of the generalized model can be implemented inside
  the classic relational model using the ordinary notions of relations on relation
  schemes and additional operations with relations.
\end{remark}

%%%%%%%%%%%%%%%%%%%%%%%%%%%%%%%%%%%%%%%%%%%%%%%%%%%%%%%%%%%%%%%%%%%%%%%%%%%%%%%%
\subsection{Attributes, Types, and Ranked Data Tables}
In this section, we present our counterpart to the classic relations on
relation schemes. We utilize the following notions.
We denote by $Y$ a (infinite denumerable) set of \emph{attributes},
any finite subset $R \subseteq Y$
is called a \emph{relation scheme.} For each attribute $y \in Y$ we consider
its \emph{type $D_y$} which is understood as the admissible set of values of
the attribute $y$, see~\cite{DaDa:3rdM} (note that in earlier literature,
types are called domains, cf.~\cite{Co:Armodflsdb}).
In the paper, we do not refer to types explicitly,
i.e., whenever we introduce an attribute,
we tacitly consider its type and for simplicity we assume
that attributes with the same name have the same type.

We utilize the usual set-theoretic representation of tuples:
A direct product $\prod_{y \in R}D_r$ of an $R$-indexed system
$\{D_y \,|\, y \in R\}$ is a set of all maps
\begin{align}
  r\!: \textstyle R \to \bigcup_{y \in R}D_y
\end{align}
such that $r(y) \in D_y$ for each $y \in R$. If $R \subseteq Y$ is finite,
then each $r \in \prod_{y \in R}D_y$ is called a \emph{tuple}
on relation scheme $R$, $r(y)$ is called the \emph{$y$-value of $r$.}
For brevity, $\prod_{y \in R}D_y$ is denoted by $\Tupl{R}$.
For $S \subseteq R$ and $r \in \Tupl{R}$, we denote by $r(S)$
the \emph{projection} of $r$ onto $S$, i.e., $r(S) \subseteq r$
such that $\langle y,d\rangle \in r(S)$ for some $d \in D_y$
if{}f $y \in S$.
In particular, $r(\emptyset) \in \Tupl{\emptyset} = \{\emptyset\}$,
i.e., $\emptyset$ is the only tuple on the empty relation scheme.
Moreover, if $r \in \Tupl{R}$, $s \in \Tupl{S}$, and
$r(R \cap S) = s(R \cap S)$, we call the set-theoretic union
$r \cup s$ the \emph{join} of tuples $r$ and $s$ and denote it by $rs$.

The relations (on relation scheme $R$) which appear in the
classic model are finite subsets of $\Tupl{R}$. Technically, such
subsets can be identified with \emph{indicator functions} which assign $1$
to finitely many tuples from $\Tupl{R}$ (to those belonging to the
relation) and $0$ otherwise. Our counterpart to relations on relation
schemes result by considering such indicator functions with codomains
being the set of degrees from complete residuated lattices.

\begin{definition}\upshape%
  Let $\mathbf{L}$ be a complete residuated lattice, $R$ be a relation scheme.
  A \emph{ranked data table} on relation scheme (shortly, an RDT) is any map
  of the form $\mathcal{D}\!: \Tupl{R} \to L$ such that
  $\{r \in \Tupl{R} \,|\, \mathcal{D}(r) > 0\}$,
  called the \emph{answer set} of $\mathcal{D}$, is finite.
  The degree $\mathcal{D}(r)$ is called
  the \emph{score} of $r$ in $\mathcal{D}$.
\end{definition}

\begin{remark}\label{rem:dee}
  (a)
  Important special cases of RDTs are represented by RDTs on
  the \emph{empty relation scheme.}
  Recall that in the classic model~\cite{DaDa:3rdM},
  there are only two relations on $\emptyset$,
  namely the empty relation on $\emptyset$
  (called \verb|TABLE_DUM| in~\cite{DaDa:3rdM})
  and the relation on $\emptyset$ containing
  the empty tuple (called \verb|TABLE_DEE|).
  In our case, all RDTs on the empty scheme are maps of the form
  $\mathcal{D}\!: \{\emptyset\} \to L$, i.e., they are uniquely given
  by the degree $\mathcal{D}(\emptyset) \in L$, i.e., by the degree which is
  assigned to $\emptyset$ (the empty tuple) by $\mathcal{D}$. Because of
  this correspondence, for each degree $a \in L$, we define
  $\Dee{a}\!: \Tupl{\emptyset} \to L$ as the RDT such
  that $\Dee{a}(\emptyset) = a$. Hence, in addition to \verb|TABLE_DUM|
  ($\Dee{0}$ in our notation) and \verb|TABLE_DEE| ($\Dee{1}$ in our notation)
  our model admits general \verb|DEE|-like RDTs for every $a \in L$, leaving
  $\Dee{0}$ and $\Dee{1}$ as two borderline cases. As it is argued
  in~\cite{DaDe:DErd}, special cases of divisions which involve
  \verb|TABLE_DUM| and \verb|TABLE_DEE| are important and have been often
  neglected in various approaches to division, which in consequence led to
  divisions with undesirable properties. In our case, the DEE-like tables
  $\Dee{a}$ play analogous important role and shall be taken into account.

  (b)
  RDTs on non-empty relation schemes can be depicted analogously as
  classic relations on non-empty relation schemes by
  two-dimensional data tables with columns
  corresponding to attributes and rows corresponding to tuples. In addition,
  each row in the table is annotated by the score of the tuple represented by
  the row (tuples with zero scores are not shown in the table).

  (c)
  If $\mathcal{D}(r) \in \{0,1\}$ for all $r \in \Tupl{R}$,
  we call $\mathcal{D}$ \emph{non-ranked.} Clearly, non-ranked RDTs are in
  a one-to-one correspondence with (finite) relations on relation schemes in
  the usual sense. A particular case of a non-ranked table is $0_R$ called
  the \emph{empty} table and satisfying $0_R(r) = 0$ for all
  $r \in \Tupl{R}$.
\end{remark}

%%%%%%%%%%%%%%%%%%%%%%%%%%%%%%%%%%%%%%%%%%%%%%%%%%%%%%%%%%%%%%%%%%%%%%%%%%%%%%%%
\subsection{Relational Operations}
By virtue of the close connection to logics based on residuated structures
of degrees, the rank-aware model we consider admits two basic types of
domain independent query systems~\cite{BeVy:TODS}.
First, a system based on evaluating predicate formulas. Second, a system
consisting of relational operations which has the same expressive power as
the former one. The relational divisions considered in this paper are
particular (fundamental or derived) relational operations. In this subsection,
we recall a fragment of the relational operations we need to cope
with divisions.

For $\mathcal{D}_1$ and $\mathcal{D}_2$ on the same relation scheme $R$,
we define
$\mathcal{D}_1 \cap \mathcal{D}_2$ (\emph{intersection}) and
$\mathcal{D}_1 \cup \mathcal{D}_2$ (\emph{union}) by
\begin{align}
  (\mathcal{D}_1 \cap \mathcal{D}_2)(r) &=
  \mathcal{D}_1(r) \wedge \mathcal{D}_2(r), \\
  (\mathcal{D}_1 \cup \mathcal{D}_2)(r) &=
  \mathcal{D}_1(r) \vee \mathcal{D}_2(r),
\end{align}
for all $r \in \mathrm{Tupl(R)}$.
In words, $\cap$ and $\cup$ are defined componentwise using
the lattice operations $\wedge$ and $\vee$ in $\mathbf{L}$.

The natural join in our model is introduced as follows.
If $\mathcal{D}_1$ is an RDT on relation scheme $R \cup S$ and
$\mathcal{D}_2$ is an RDT of relation scheme $S \cup T$ such that
$R \cap S = R \cap T = S \cap T = \emptyset$
(i.e., $R$, $S$, and $T$ are pairwise disjoint), then the
\emph{natural join} of $\mathcal{D}_1$ and $\mathcal{D}_2$ is an RDT
on relation scheme $R \cup S \cup T$
denoted by $\join{}{\mathcal{D}_1}{\mathcal{D}_2}$ and defined by
\begin{align}
  \bigl(\join{}{\mathcal{D}_1}{\mathcal{D}_2}\bigr)(rst) &=
  \mathcal{D}_1(rs) \otimes \mathcal{D}_2(st),
  \label{eqn:natjoin}
\end{align}
for each $r \in \Tupl{R}$, $s \in \Tupl{S}$,
and $t \in \Tupl{T}$. Hence, $\otimes$ in $\mathbf{L}$ acts as a conjunctive
aggregator which generalizes the classic conjunction appearing in the definition
of ordinary natural join of relations.
If $\mathcal{D}$ is an RDT on $R$,
the \emph{projection} of $\mathcal{D}$ onto $S \subseteq R$ is
denoted by $\project{S}{\mathcal{D}}$ and defined by
\begin{align}
  (\project{S}{\mathcal{D}})(s) &=
  \textstyle\bigvee_{\!t \in \Tupl{R \setminus S}}\mathcal{D}(st),
  \label{def:project}
\end{align}
for each $s \in \Tupl{S}$. Using projections of tuples onto $S$,
we may write~\eqref{def:project} equivalently as
$(\project{S}{\mathcal{D}})(s) =
\bigvee\{\mathcal{D}(r) \,|\, r(S) = s\}$.
Since $\otimes$ is distributive over~$\bigvee$,
we may introduce a \emph{semijoin} of $\mathcal{D}_1$ on $R$
and $\mathcal{D}_2$ on $S$ as
$\project{R}{\join{}{\mathcal{D}_1}{\mathcal{D}_2}}$
or equivalently as
$\join{}{\mathcal{D}_1}{\project{R \cap S}{\mathcal{D}_2}}$
and we denote it $\mathcal{D}_1 \ltimes \mathcal{D}_2$.

Analogously as in the classic model, semijoins in our model are
important since they allow us to algebraically express existential
queries of the form of categorical propositions ``some $\varphi$ is $\psi$''
or, in the database terminology~\cite{DaDa:3rdM},
``some tuples from $\mathcal{D}_1$ are
\emph{matching} tuples in $\mathcal{D}_2$''.

\begin{remark}
  (a)
  One may check that if all arguments to the above-mentioned operations
  are non-ranked, then the results of relational operations coincide with
  the results of the classic relational operations of union, intersection,
  natural join, and projection~\cite{Co:Armodflsdb,DaDa:3rdM}.

  (b)
  Let us comment on the role of the general suprema in~\eqref{def:project}.
  In predicate logics based on residuated structures of
  degrees~\cite{CiHa:Tnbpfl}, general suprema are used to interpret
  existentially quantified formulas. In a more detail, for a formula of the
  form $(\exists x)\varphi$, its truth degree
  $||(\exists x)\varphi||_{\mathbf{M},v}$ in
  the $\mathbf{L}$-structure $\mathbf{M}$ under the evaluation $v$ of
  object variables is defined as the \emph{supremum} of all truth degrees
  $||\varphi||_{\mathbf{M},w}$ where $w(y) = v(y)$ for each
  variable $y$ such that $y \ne x$. Put in words,
  $||(\exists x)\varphi||_{\mathbf{M},v}$ is the least upper bound of all
  degrees to which $\varphi$ is true in $\mathbf{M}$ considering $x$ as
  a variable which can be assigned any value from the universe of $\mathbf{M}$.
  Note that if $\mathbf{L}$ is the two-element Boolean algebra, this
  interpretation coincides exactly with the usual interpretation of
  existentially quantified formulas and, in particular,
  $||(\exists x)\varphi||_{\mathbf{M},v} = 1$ if{}f there is $w$
  such that $||\varphi||_{\mathbf{M},w} = 1$ and $w(y) = v(y)$
  for all $y \ne x$ (i.e., $x$ can be assigned a value which
  makes $\varphi$ true in $\mathbf{M}$). Now, since projections are
  relational operations which express queries formulated by existentially
  quantified formulas in relational calculi, \eqref{def:project}
  is defined in terms of~$\bigvee$. In words, $(\project{S}{\mathcal{D}})(s)$
  is a degree to which \emph{there is} a tuple in $\mathcal{D}$ whose projection
  onto $S$ equals to $s$.
\end{remark}

%%%%%%%%%%%%%%%%%%%%%%%%%%%%%%%%%%%%%%%%%%%%%%%%%%%%%%%%%%%%%%%%%%%%%%%%%%%%%%%%
%%%%%   EXISTING AND NEW APPROACHES TO DIVISION
%%%%%%%%%%%%%%%%%%%%%%%%%%%%%%%%%%%%%%%%%%%%%%%%%%%%%%%%%%%%%%%%%%%%%%%%%%%%%%%%
\section{Existing and New Approaches to Division}\label{sec:appr}
In this section, we review several classic approaches to division which
appeared in the literature on database systems, present their rank-aware
counterparts, and comment on their relationship to the existing rank-aware
or fuzzy approaches in databases. The section is structured into subsections
which roughly follow the structure of~\cite{DaDe:DErd} which is arguably
the best comparison of division-like operations from the point of view of the
relational model of data.

In this section, whenever we say that (a relation or an RDT) $\mathcal{D}$
is on scheme $RS$, we mean that it is defined on the scheme $R \cup S$
such that $R \cap S = \emptyset$.

%%%%%%%%%%%%%%%%%%%%%%%%%%%%%%%%%%%%%%%%%%%%%%%%%%%%%%%%%%%%%%%%%%%%%%%%%%%%%%%%
\subsection{Codd-style Division}
Historically, the Codd division is the initial operation in the family of
division-like operations. Its initial purpose was technical---to ensure
completeness of the relational algebra with respect to the relational
calculus which allows us to express queries involving universal quantification.
Strictly speaking, its presence in the relational algebra is not necessary
since in the classical logic, universally quantified formulas of the from
$(\forall x)\varphi$ can be replaced by formulas 
$\blneg(\exists x)\blneg\varphi$,
i.e., universal quantifiers are expressible by means of negations and
existential quantification. Thus, the division is considered as
a derived operation which is expressed by means of set-theoretic
difference (relational counterparts to negations) and projections
(relational counterparts to existential quantification).

Namely, for a relation $\mathcal{D}_1$ on $RS$ and relation
$\mathcal{D}_2$ on $S$, the Codd division
$\mathcal{D}_1 \div_\mathrm{Codd} \mathcal{D}_2$
may be introduced~\cite{DaDe:DErd} as
\begin{align}
  \mathcal{D}_1 \div_\mathrm{Codd} \mathcal{D}_2
  &=
  \project{R}{\mathcal{D}_1} \setminus
  \project{R}{(\join{}{\project{R}{\mathcal{D}_1}}{\mathcal{D}_2})
    \setminus \mathcal{D}_1},
  \label{eqn:Codd-}
\end{align}
where $\pi_R$, $\bowtie$, and $\setminus$ denote the usual projection,
natural join (cross join in this particular case), and set-theoretic difference,
respectively.
The survey chapter~\cite{DaDe:DErd} identifies several epistemic issues
of~\eqref{eqn:Codd-}. The most important are:
\begin{enumerate}
\item[\itm{1}]
  Unlike semijoins, \eqref{eqn:Codd-} is restricted to relations on particular
  schemes, i.e., the operation cannot be performed with relations on arbitrary
  schemes which makes it less general (and less useful).
\item[\itm{2}]
  The meaning of~\eqref{eqn:Codd-} does not faithfully correspond to the
  categorical proposition ``all $\varphi$ are $\psi$''. If $\varphi$
  is $s \in \mathcal{D}_2$ and $\psi$ is $rs \in \mathcal{D}_1$, then
  \begin{align}
    (\forall s)(s \in \mathcal{D}_2 \blRightarrow rs \in \mathcal{D}_1)
    \label{eqn:foralls}
  \end{align}
  is true for all $r \in \Tupl{R}$ provided that $\mathcal{D}_2$ is empty.
  In contrast, the result of~\eqref{eqn:Codd-} is always a subset of
  $\project{R}{\mathcal{D}_1}$. Hence, in general,
  the meaning of~\eqref{eqn:Codd-} is
  ``any \emph{$r$ in $\project{R}{\mathcal{D}_1}$}
  such that $rs \in \mathcal{D}_1$
  for all $s \in \mathcal{D}_2$'' rather than ``any $r$ such that
  $rs \in \mathcal{D}_1$ for all $s \in \mathcal{D}_2$'', cf.~\cite{DaDe:DErd}.
  As a consequence, \eqref{eqn:Codd-} is equivalent to
  \begin{align}
    \mathcal{D}_1 \div \mathcal{D}_2 &=
    \bigl\{
    r \in \project{R}{\mathcal{D}_1} \,|\,
    \text{for all } s \in \mathcal{D}_2
    \text{, we have }
    rs \in \mathcal{D}_1
    \bigr\},
    \label{eqn:with_range}
  \end{align}
  where $\project{R}{\mathcal{D}_1}$
  can be seen as the \emph{range} for the division.
\end{enumerate}

By a direct generalization of~\eqref{eqn:Codd-} in rank-aware approaches,
we inherit both the issues. In addition, it is questionable how to handle
$\setminus$ in the presence of scores. One way to go is to consider
$(\mathcal{D}_1 \setminus \mathcal{D}_2)(r)$ to be the degree to which
$r$ is in $\mathcal{D}_1$ and \emph{is not} in $\mathcal{D}_2$ and express the
negation using $\rightarrow$ and $0$, i.e.,
\begin{align}
  (\mathcal{D}_1 \setminus \mathcal{D}_2)(r) &=
  \mathcal{D}_1(r) \otimes (\mathcal{D}_2(r) \rightarrow 0).
  \label{eqn:rminus}
\end{align}
Although $\mathcal{D}_1 \setminus \mathcal{D}_2$ is always finite,
it does not fulfill basic properties one would expect for a difference.
For instance, $\mathcal{D}_1 \setminus \mathcal{D}_2 = 0_R$
does not imply $\mathcal{D}_1 \subseteq \mathcal{D}_2$ in general.
Alternatively, one may introduce $\setminus$ as an independent
fundamental connective in $\mathbf{L}$ and induce the difference of
RDTs componentwise analogously as $\cap$ or $\cup$. For instance, one
may use commutative doubly-residuated lattices~\cite{OrRa:Rmcs}
with $\setminus$ being adjoint to a non-idempotent disjunction.
Note that difference-like operations with relations (with scores) in
the database literature are often defined analogously
as~\eqref{eqn:rminus}, usually on $L=[0,1]$ with $\otimes$ being
the minimum and $\rightarrow$ being
the \L ukasiewicz implication~\cite{PiBo:book}.
The general issue with graded style-versions of \eqref{eqn:Codd-}
is that universal quantifier (interpreted by infima in $\mathbf{L}$)
is not definable using the existential one (interpreted by suprema
in $\mathbf{L}$).

Most common truth-functional
approaches~\cite{Bo:Opdfr,BoPiRo:Crdfr,BoDuPiPr:Fqrdedo,DuPr:Sqofrd}
that can be found in literature
on rank-aware extensions generalize~\eqref{eqn:foralls} by putting
\begin{align}
  (\mathcal{D}_1 \div \mathcal{D}_2)(r) &=
  \textstyle\bigwedge_{s \in \Tupl{S}}
  \bigl(\mathcal{D}_2(s) \rightarrow \mathcal{D}_1(rs)\bigr)
  \label{eqn:gradeCodd_dd}
\end{align}
for all $r \in \mathrm{Tupl}(R)$ provided that $\mathcal{D}_1$ and
$\mathcal{D}_2$ are RDTs on schemes $RS$ and $S$, respectively.
In our setting, $\bigwedge$ is the operation of infimum in $\mathbf{L}$,
and $\rightarrow$ is the residuum in $\mathbf{L}$. The above-cited
approaches often use a fixed scale of degrees (with $L=[0,1]$)
with $\rightarrow$ being a general truth function of implication.
In addition to $\rightarrow$ which
are adjoint to $\otimes$ (so-called R-implications),
the approaches use S-implications~\cite{BoPiRo:Crdfr}.
We do not want to endorse this concept here because of its marginal role
in fuzzy logics in the \emph{narrow sense}, see~\cite{Got:Mfl} and the
soundness issues regarding S-implications.

\begin{remark}
  Observe that since $r \in \mathrm{Tupl}(R)$, \eqref{eqn:gradeCodd_dd}
  solves issue \itm{2} but this is at the expense of
  losing domain independence.
  Indeed, if $R$ contains an attribute which has a type consisting of
  infinitely many values then the result $\mathcal{D}_1 \div \mathcal{D}_2$
  defined by \eqref{eqn:gradeCodd_dd} is \emph{infinite} which is highly undesirable
  property from the database viewpoint---if a materialization of
  $\mathcal{D}_1 \div \mathcal{D}_2$ is necessary in order to evaluate
  a compound query involving the division, the evaluation cannot
  be performed (in finitely many steps). Probably because of this issue,
  some of the graded approaches cited above use~\eqref{eqn:gradeCodd_dd}
  assuming that $(\project{R}{\mathcal{D}_1})(r) > 0$ which,
  unfortunately, introduces \itm{2} again.
\end{remark}

In our previous work~\cite{BeVy:Qssbd}, we have used a fundamental
domain-dependent division operation which is sufficient to establish
the equivalence between a domain-dependent relational algebra and
a domain relational calculus.
Recently, we have proposed a domain independent variant~\cite{BeVy:TODS}
with explicit \emph{range} which is used to establish the equivalence between
a domain-independent relational algebra and a domain relational
calculus with range declarations.
The operation is defined as follows.

Let $\mathcal{D}_1$, $\mathcal{D}_2$, and $\mathcal{D}_3$ be RDTs
on $RS$, $S$, and $R$, respectively.
Then, a \emph{division} $\rdiv{\mathcal{D}_3}{\mathcal{D}_1}{\mathcal{D}_2}$
of $\mathcal{D}_1$ by $\mathcal{D}_2$ which ranges over $\mathcal{D}_3$ is
an RDT on $R$ defined by
\begin{align}
  \bigl(\rdiv{\mathcal{D}_3}{\mathcal{D}_1}{\mathcal{D}_2}\bigr)(r) =
  \textstyle\bigwedge_{s \in \mathrm{Tupl}(S)}
  \bigl(
  \mathcal{D}_3(r) \otimes (\mathcal{D}_2(s) \rightarrow  \mathcal{D}_1(rs))
  \bigr),
  \label{def:rdiv}
\end{align}
for each $r \in \mathrm{Tupl}(R)$.
Clearly, $\rdiv{\mathcal{D}_3}{\mathcal{D}_1}{\mathcal{D}_2} \subseteq
\mathcal{D}_3$. In addition, \eqref{def:rdiv} possesses further desirable
properties. For instance, if $\mathcal{D}_3$ is non-ranked,
then $\rdiv{\mathcal{D}_3}{\mathcal{D}_1}{\mathcal{D}_2}$
is the greatest among all $\mathcal{D} \subseteq \mathcal{D}_3$ such that
$\join{}{\mathcal{D}}{\mathcal{D}_2} \subseteq \mathcal{D}_1$.
Furthermore, if $R = \emptyset$ and $\mathcal{D}_3 = \Dee{1}$
(see Remark~\ref{rem:dee}), then \eqref{def:rdiv} becomes
(the relational representation of) the subsethood degree
of $\mathcal{D}_2$ in $\mathcal{D}_1$, see~\cite{Bel:FRS}.
Also, the definition eliminates \itm{2} and is domain independent.

%%%%%%%%%%%%%%%%%%%%%%%%%%%%%%%%%%%%%%%%%%%%%%%%%%%%%%%%%%%%%%%%%%%%%%%%%%%%%%%%
\subsection{Date's Small Divide (Original and Generalized)}
In order to overcome issue \itm{2}, Date (see~\cite{DaDe:DErd} and the
references therein) proposed a Small Divide operation. Consider the
following relations on relation schemes:
$\mathcal{D}_1$ on $R$ (called the \emph{dividend}),
$\mathcal{D}_2$ on $S$ (called the \emph{divisor}),
$\mathcal{D}_3$ on $RS$ (called the \emph{mediator}).
Then, the original version of
Small Divide~\cite{DaDe:DErd} is
\begin{align}
  \mathcal{D}_1 \div^{\mathcal{D}_3}_\mathrm{sdo} \mathcal{D}_2
  &=
  \mathcal{D}_1 \setminus
  \project{R}{(\join{}{\mathcal{D}_1}{\mathcal{D}_2})
    \setminus \mathcal{D}_3} \notag \\
  &=
  \bigl\{r \in \mathcal{D}_1 \,|\,
  \text{for all } s \in \mathcal{D}_2
  \text{, we have }
  rs \in \mathcal{D}_3
  \bigr\}.
  \label{eqn:SmallOriginal}
\end{align}
A graded generalization
of~\eqref{eqn:SmallOriginal} is
\begin{align}
  \bigl(
  \mathcal{D}_1 \div^{\mathcal{D}_3}_\mathrm{gsdo} \mathcal{D}_2
  \bigr)(r) =
  \mathcal{D}_1(r) \otimes
  \textstyle\bigwedge_{s \in \mathrm{Tupl}(S)}
  \bigl(
  \mathcal{D}_2(s) \rightarrow  \mathcal{D}_3(rs)\bigr)
  \label{def:gSmallOriginal}
\end{align}
with $\mathcal{D}_1$, $\mathcal{D}_2$, and $\mathcal{D}_3$ being RDTs
on $R$, $S$, and $RS$, respectively. The graded variant of the Small
Divide and~\eqref{def:rdiv} are equivalent under the following conditions:

\begin{theorem}
  If\/ $\mathbf{L}$ is prelinear or divisible, then
  $\rdiv{\mathcal{D}_1}{\mathcal{D}_3}{\mathcal{D}_2} =
  \mathcal{D}_1 \div^{\mathcal{D}_3}_\mathrm{gsdo} \mathcal{D}_2$.
\end{theorem}
\begin{proof}
  Either of prelinearity or divisibility ensures that
  $a \otimes (b \wedge c) = (a \otimes b) \wedge (a \otimes c)$
  for all $a,b,c \in L$, see~\cite{Bel:FRS,EsGo:MTL,Haj:MFL}.
  In addition, since $\mathcal{D}_2$ and $\mathcal{D}_3$
  are finite, in both~\eqref{def:gSmallOriginal} and~\eqref{def:rdiv} the
  infimum is computed using only finitely many degrees other than $1$, i.e.,
  the claim follows by distributivity of $\otimes$ over infima of
  finitely many degrees which are pairwise distinct.
\end{proof}

\noindent
Note that analogous observation holds if $\mathbf{L}$ is arbitrary
and $\mathcal{D}_1$ is non-ranked.

\begin{remark}
  The previous observation has two important consequences: In the mainstream
  fuzzy logics (based on prelinear residuated lattices), graded Small Divide
  and~\eqref{def:rdiv} are equivalent. In particular, if $\mathbf{L}$ is the
  two-element Boolean algebra, the ranked model becomes the classic one, i.e.,
  this observation pertains to the classic relational model.
\end{remark}

In order to cope with issue \itm{1}, the original Small Divide has been
further extended to accomodate relations on more general schemes. Namely,
for $\mathcal{D}_1$ on $RT$, $\mathcal{D}_2$ on $SU$, and
$\mathcal{D}_3$ on $RSV$, Date introduced~\cite{DaDe:DErd}
a general form of Small Divide as follows:
\begin{align}
  \mathcal{D}_1 \div^{\mathcal{D}_3}_\mathrm{sd} \mathcal{D}_2
  &=
  \mathcal{D}_1 \nltimes
  ((\join{}{\project{R}{\mathcal{D}_1}}{\project{S}{\mathcal{D}_2}})
    \nltimes \mathcal{D}_3).
    \label{eqn:Small}
\end{align}
where $\nltimes$ denotes the \emph{semidifference}, i.e.,
$\mathcal{D} \nltimes \mathcal{D}' =
\mathcal{D} \setminus (\mathcal{D} \ltimes \mathcal{D}')$.
By moment's reflection, we derive that
\begin{align}
  \mathcal{D}_1 \div^{\mathcal{D}_3}_\mathrm{sd} \mathcal{D}_2
  &=
  \bigl\{rt \in \mathcal{D}_1 \,|\,
  \text{for all } s \in \project{S}{\mathcal{D}_2}
  \text{, we have }
  rs \in \project{RS}{\mathcal{D}_3}
  \bigr\}.
\end{align}
We may therefore introduce the following operation in the graded setting
\begin{align}
  \bigl(
  \mathcal{D}_1 \div^{\mathcal{D}_3}_\mathrm{gsd} \mathcal{D}_2
  \bigr)(rt) =
  \mathcal{D}_1(rt) \otimes
  \textstyle\bigwedge_{s \in \mathrm{Tupl}(S)}
  \bigl(
  (\project{S}{\mathcal{D}_2})(s) \rightarrow
  (\project{RS}{\mathcal{D}_3})(rs)\bigr)
  \label{def:gSmall}
\end{align}
provided that $\mathcal{D}_1$, $\mathcal{D}_2$, and $\mathcal{D}_3$ are RDTs
on $RT$, $SU$, $RSV$, respectively. As in the classic setting,
$\div_\mathrm{gsd}$ eliminates both the issues \itm{1} and \itm{2}
mentioned earlier.

%%%%%%%%%%%%%%%%%%%%%%%%%%%%%%%%%%%%%%%%%%%%%%%%%%%%%%%%%%%%%%%%%%%%%%%%%%%%%%%%
\subsection{Todd-style Division}
An alternative approach to eliminate issue \itm{1} is the division proposed
by Todd, cf.~\cite{DaDe:DErd}. Written directly in the set notation,
\begin{align}
  \mathcal{D}_1 \div_\mathrm{Todd} \mathcal{D}_2 &=
  \{rt \in \mathcal{U}
  \,|\,
  \text{for all } s \in \Tupl{S}
  \text{: if }
  st \in \mathcal{D}_2
  \text{, then }
  rs \in \mathcal{D}_1\},
  \label{eqn:Todd}
\end{align}
where
$\mathcal{U} =
\join{}{\project{R}{\mathcal{D}_1}}{\project{T}{\mathcal{D}_2}}$.
Unfortunately, $\div_\mathrm{Todd}$ and its direct rank-aware generalizations
inherit the issue~\itm{2}. This is caused by the fact that the ranges for
$r$ and $t$ in~\eqref{eqn:Todd} are considered to be the projections of
$\mathcal{D}_1$ and $\mathcal{D}_2$, respectively. Interestingly,
if $\mathcal{U}$ is considered to be the set of all tuples on $RT$,
the graded generalization becomes
\begin{align}
  (\mathcal{D}_1 \div_\mathrm{gTodd} \mathcal{D}_2)(rt)
  &=
  \textstyle\bigwedge_{s \in \Tupl{S}}
  \bigl(\mathcal{D}_2(st) \rightarrow \mathcal{D}_1(rs)\bigr)
  \label{eqn:KB}
\end{align}
which is the Kohout-Bandler superproduct composition~\cite{BaKo:Sipfrp,Bel:FRS}
of fuzzy relations $\mathcal{D}_1$ and $\mathcal{D}_2$ (in this order).
As in the case of~\eqref{eqn:gradeCodd_dd},
$\div_\mathrm{gTodd}$ is domain dependent, i.e., even if $\mathcal{D}_1$
and $\mathcal{D}_2$ are finite, the result of~\eqref{eqn:KB} may be infinite
which is an undesirable property.
%  We stop our review of Todd-style divisions
% here and turn directly to its generalization, the Great Divide.

%%%%%%%%%%%%%%%%%%%%%%%%%%%%%%%%%%%%%%%%%%%%%%%%%%%%%%%%%%%%%%%%%%%%%%%%%%%%%%%%
\subsection{Date's Great Divide}
In the same spirit as the Small Divide has been proposed to eliminate the
issues of the classic Codd division, the Great Divide has been proposed by
Date~\cite{DaDe:DErd} to deal with the issues of the Todd division. Again,
we may assume two variants of the operation---the original one and
the generalized one. For illustration, we focus here only on the original
variant, the generalized one can be obtained in much the same way as in
the case of the Small Divide.

According to~\cite{DaDe:DErd}, for relations
$\mathcal{D}_1$ on $R$ (called the \emph{dividend}),
$\mathcal{D}_2$ on $T$ (called the \emph{divisor}),
$\mathcal{D}_3$ on $RS$ (called the \emph{first mediator}),
and $\mathcal{D}_4$ on $ST$ (called the \emph{second mediator}),
we put
\begin{align}
  \mathcal{D}_1 \div^{\mathcal{D}_3,\mathcal{D}_4}_\mathrm{gdo} \mathcal{D}_2
  &=
  (\join{}{\mathcal{D}_1}{\mathcal{D}_2}) \nltimes
  ((\join{}{\mathcal{D}_1}{\mathcal{D}_4})
    \nltimes \mathcal{D}_3).
    \label{eqn:GreatOriginal}
\end{align}
The definition~\eqref{eqn:GreatOriginal} can be equivalently expressed
in the set notation as follows:
\begin{align}
  \mathcal{D}_1 \div^{\mathcal{D}_3,\mathcal{D}_4}_\mathrm{gdo} \mathcal{D}_2
  &=
  \bigl\{rt \in \mathcal{U} \,|\,
  \text{for all } s \in \Tupl{S}
  \text{: if }
  st \in \mathcal{D}_4
  \text{, then }
  rs \in \mathcal{D}_3
  \bigr\},
  \label{eqn:GreatOriginal_set}
\end{align}
where $\mathcal{U} = \join{}{\mathcal{D}_1}{\mathcal{D}_2}$.
Based on~\eqref{eqn:GreatOriginal_set}, we may introduce a graded variant
$\div_\mathrm{ggdo}$ of the original Great Divide as follows
\begin{align}
  \bigl(
  \mathcal{D}_1
  \div^{\mathcal{D}_3,\mathcal{D}_4}_\mathrm{ggdo}
  \mathcal{D}_2
  \bigr)(rt) &=
  \mathcal{D}_1(r) \otimes \mathcal{D}_2(t) \otimes
  \textstyle\bigwedge_{s \in \mathrm{Tupl}(S)}
  \bigl(
  \mathcal{D}_4(st) \rightarrow \mathcal{D}_3(rs)\bigr) \notag \\
  &=
  (\join{}{\mathcal{D}_1}{\mathcal{D}_2})(rt)
  \otimes
  \textstyle\bigwedge_{s \in \mathrm{Tupl}(S)}
  \bigl(
  \mathcal{D}_4(st) \rightarrow \mathcal{D}_3(rs)\bigr)
  \label{def:gGreatOriginal}
\end{align}
with $\mathcal{D}_1$, $\mathcal{D}_2$, $\mathcal{D}_3$,
and $\mathcal{D}_4$ being RDTs on $R$, $T$, $RS$, and $ST$, respectively.
Loosely speaking, \eqref{def:gGreatOriginal} can be seen as
a domain-independent variant of the Kohout-Bandler superproduct composition
whose range is limited to the natural join of $\mathcal{D}_1$
and $\mathcal{D}_2$.

Analogously as in the classic case, the graded Great Divide is more
general than the graded Small Divide. In particular, $\div_\mathrm{gsdo}$
can be seen as $\div_\mathrm{ggdo}$ with the divisior being the RDT $\Dee{1}$
on the empty relation scheme:

\begin{corollary}
  We have
  $\mathcal{D}_1 \div_\mathrm{gsdo}^{\mathcal{D}_3} \mathcal{D}_2 =
  \mathcal{D}_1 \div_\mathrm{ggdo}^{\mathcal{D}_3,\mathcal{D}_2} \Dee{1}$.
  \qed
\end{corollary}

As we have already mentioned, \eqref{def:gGreatOriginal} can be generalized
in a similar way as~\eqref{def:gSmall} to handle RDTs on more general
relational schemes.

%%%%%%%%%%%%%%%%%%%%%%%%%%%%%%%%%%%%%%%%%%%%%%%%%%%%%%%%%%%%%%%%%%%%%%%%%%%%%%%%
\subsection{Darwen's Divide}
Later, Darwen \cite{DaDe:DErd} proposed another division-like
operation which is now commonly called Darwen's Divide.
This operation is defined similarly as Date's Great Divide
but it does not impose any requirements on the relation schemes of
its arguments.

The definition is as follows \cite{DaDe:DErd}. For relations
$\mathcal{D}_1$ on $R_1$ (called the \emph{dividend}),
$\mathcal{D}_2$ on $R_2$ (called the \emph{divisor}),
$\mathcal{D}_3$ on $R_3$ (called the \emph{first mediator}),
and $\mathcal{D}_4$ on $R_4$ (called the \emph{second mediator}),
we put
\begin{align}
  \mathcal{D}_1 \div^{\mathcal{D}_3,\mathcal{D}_4}_\mathrm{ddo} \mathcal{D}_2
  &=
  (\join{}{\mathcal{D}_1}{\mathcal{D}_2}) \nltimes
  ((\join{}{\mathcal{D}_1}{\mathcal{D}_4})
    \nltimes \mathcal{D}_3).
    \label{def:DarwenOriginal}
\end{align}
Note that the relation scheme of result of Darwen's Divide is $R_1 \cup R_2$
since $R_1$ and $R_2$ are arbitrary and might have some attributes in common.

In the proof of the set notation of Darwen's Divide we utilize the following
lemma.

\begin{lemma}\label{lem:SemidiffChar}
  Consider relations $\mathcal{D}_1$ on $RS$ and $\mathcal{D}_2$ on $ST$ such that
  $R, S, T$ are pairwise disjoint ($R \cap S = R \cap T = S \cap T = \emptyset$).
  For every tuple $r \in \Tupl{R}$ and $s \in \Tupl{S}$ we have
  $rs \in \mathcal{D}_{1} \nltimes \mathcal{D}_{2}$ if{}f
  \begin{align}
    rs \in \mathcal{D}_{1}
    \blwedge
    \blneg (\exists t \in \Tupl{T}) st \in \mathcal{D}_{2}
    \label{eqn:SemidiffChar}
  \end{align}
  or equivalently
  \begin{align}
    rs \in \mathcal{D}_{1}
    \blwedge
    \blneg (\exists s't \in \Tupl{ST})
    \left(
    	(rs)(S) = (s't)(S) \blwedge s't \in \mathcal{D}_{2}
    \right),
    \label{eqn:SemidiffCharEquiv}
  \end{align}
  where $s' \in \Tupl{S}$ and $t \in \Tupl{T}.$
\begin{proof}
The first part follows directly from the definition of semidifference:
\begin{align*}
rs \in \mathcal{D}_{1} \nltimes \mathcal{D}_{2} \iff \enspace
& rs \in \mathcal{D}_{1}
  \setminus \pi_{RS}
  ( \mathcal{D}_{1} \bowtie \mathcal{D}_{2} ) 
%\\ \iff \enspace &
%rs \in \mathcal{D}_{1}
%\wedge \neg (rs \in \pi_{RS} ( \mathcal{D}_{1} \bowtie \mathcal{D}_{2} ))
%\\ \iff \enspace &
%rs \in \mathcal{D}_{1}
%\wedge \neg (rs \in (\mathcal{D}_{1} \bowtie \pi_{S}(\mathcal{D}_{2})))
\\ \iff \enspace &
rs \in \mathcal{D}_{1}
\blwedge
\blneg (rs \in \mathcal{D}_{1} \blwedge s \in \pi_{S}(\mathcal{D}_{2}))
%\\ \iff \enspace &
%rs \in \mathcal{D}_{1}
%\wedge
%( \neg (rs \in \mathcal{D}_{1}) \vee \neg (s \in \pi_{S}(\mathcal{D}_{2})))
\\ \iff \enspace &
\underbrace{(rs \in \mathcal{D}_{1}
  \blwedge
  \blneg rs \in \mathcal{D}_{1})}_{\text{always false}
}
\blvee
(rs \in \mathcal{D}_{1} \blwedge \blneg s \in \pi_{S}(\mathcal{D}_{2}))
%\\ \iff \enspace &
%rs \in \mathcal{D}_{1} \wedge \neg (s \in \pi_{S}(\mathcal{D}_{2}))
\\ \iff \enspace &
rs \in \mathcal{D}_{1}
\blwedge
\blneg (\exists t \in \Tupl{T}) st \in \mathcal{D}_{2}.
\end{align*}
The rest follows from the fact that $R, S, T$ are pairwise disjoint and
$(rs)(S) = (s't)(S)$ is equivalent to $s = s'$.
\end{proof}
\end{lemma}

To simplify the notation, for two tuples $r_1 \in \Tupl{R_1}$ and
$r_2 \in \Tupl{R_2}$ we denote by $r_1 \between r_2$ the fact that $r_1$ and
$r_2$ are joinable ($r_1(R_1 \cap R_2) = r_2(R_1 \cap R_2)$).

Let us note that the Lemma \ref{lem:SemidiffChar} can be applied to relations on
 arbitrary schemes. For relations $\mathcal{D}_1$ on $R_1$ and
$\mathcal{D}_2$ on $R_2$ it suffices to put
$R = R_1 \setminus R_2, S = R_1 \cap R_2$ and $T = R_2 \setminus R_1$.
Obviously, relation schemes $R, S, T$ defined in this manner are pairwise
disjoint and it holds that $R_1 = R \cup S$ and $R_2 = S \cup T.$ Now for
$r_1 \in \Tupl{R_1}$ using \eqref{eqn:SemidiffCharEquiv} we have
$r_1 \in \mathcal{D}_{1} \nltimes \mathcal{D}_{2}$ if{}f
\begin{align}
  r_1 \in \mathcal{D}_{1}
  \blwedge
  \blneg (\exists r_2 \in \Tupl{R_2}) \left(
  r_1 \between r_2 \blwedge r_2 \in \mathcal{D}_{2} \right)
  \label{eqn:SemidiffCharArbitrary}
\end{align}

To put \eqref{eqn:SemidiffCharArbitrary} in words, tuple $r_1$ belongs to the
result of semidifference of $\mathcal{D}_1$ and $\mathcal{D}_2$ (in this order)
if{}f $r_1$ belongs to $\mathcal{D}_1$ and there is no tuple $r_2$ from
$\mathcal{D}_2$ that is joinable with $r_1$.

\begin{theorem}
Consider relations
$\mathcal{D}_1$ on $R_1$,
$\mathcal{D}_2$ on $R_2$,
$\mathcal{D}_3$ on $R_3$,
and $\mathcal{D}_4$ on $R_4$. The definition \eqref{def:DarwenOriginal} can be
equivalently expressed in the set notation as follows:
\begin{align}
  \mathcal{D}_1 &\div^{\mathcal{D}_3,\mathcal{D}_4}_\mathrm{ddo} \mathcal{D}_2
  \label{eqn:DarwenOriginal_set}
  = \\ \notag
  &
  \bigl\{r_1 r_2 \in \mathcal{U} \,|\,
  \text{for all } r_4 \in \mathcal{D}_4
  \text{: if }
  r_1 r_2 \between r_4
  \text{, then there is } r_3 \in \mathcal{D}_3
  \text{: }
  r_1 r_4 \between r_3
  \bigr\},
\end{align}
where $\mathcal{U} = \join{}{\mathcal{D}_1}{\mathcal{D}_2}$.

\begin{proof}
First, the fact that
$\mathcal{D}_1 \div^{\mathcal{D}_3,\mathcal{D}_4}_\mathrm{ddo} \mathcal{D}_2
\subseteq
\join{}{\mathcal{D}_1}{\mathcal{D}_2}$ follows directly from the definition of
semidifference.

For brevity, in the following proof we will denote the join of $\mathcal{D}_1$
and $\mathcal{D}_2$ by $\mathcal{U} = \join{}{\mathcal{D}_1}{\mathcal{D}_2}$.
Now, let $r_1 \in \Tupl{R_1}$ and $r_2 \in \Tupl{R_2}$ be joinable tuples.
Using \eqref{eqn:SemidiffCharArbitrary} we have
\begin{align*}
&
r_1 r_2 \in
  \mathcal{D}_1 \div^{\mathcal{D}_3,\mathcal{D}_4}_\mathrm{ddo} \mathcal{D}_2
\\ \iff \enspace &
r_1 r_2 \in
  (\mathcal{D}_{1} \bowtie \mathcal{D}_{2})
  \nltimes
  ((\mathcal{D}_{1} \bowtie \mathcal{D}_{4}) \nltimes \mathcal{D}_{3})
\\ \iff \enspace &
r_1 r_2 \in \mathcal{U}
\blwedge
\blneg (\exists r' \in \Tupl{R_1 \cup R_4}) 
	\left(
  r_1 r_2 \between r' \blwedge r' \in
    ( \mathcal{D}_{1} \bowtie \mathcal{D}_{4}) \nltimes \mathcal{D}_{3}
  	\right)
\end{align*}

The tuple $r' \in \Tupl{R_1 \cup R_4}$ can be seen as a join of tuples
$r' = r_1' r_4$, where $r_1' \in \Tupl{R_1}$ and $r_4 \in \Tupl{R_4}$ such that
$r_1' \between r_4$. We can replace the $(\exists r' \in \Tupl{R_1 \cup R_4})$
with $(\exists r_1' \in \Tupl{R_1}) (\exists r_4 \in \Tupl{R_4})$ and additional
constraint that ensures joinability of $r_1'$ and $r_4$.

It is easy to see that $r_1 r_2$ is joinable with $r'$ if and only if $r_1 r_2$
is joinable with all ``components'' of $r'$ (here with both $r_1'$ and $r_4$).
Symbolically, we have $r_1 r_2 \between r'$ if{}f
$r_1 r_2 \between r_1' \blwedge r_1 r_2 \between r_4$. Since both
$r_1, r_1' \in \Tupl{R_1}$, the first condition $r_1 r_2 \between r_1'$ is
equivalent to $r_1 = r_1'$. Furthermore, second condition $r_1 r_2 \between r_4$
 implies $r_1 \between r_4$.

Continuing the proof and applying \eqref{eqn:SemidiffCharArbitrary} to the
second semidifference we have
\begin{align*}
& r_1 r_2 \in
  \mathcal{D}_1 \div^{\mathcal{D}_3,\mathcal{D}_4}_\mathrm{ddo} \mathcal{D}_2
\\ \iff \enspace &
r_1 r_2 \in \mathcal{U}
\blwedge
\blneg (\exists r' \in \Tupl{R_1 \cup R_4}) \left( 
  r_1 r_2 \between r' \blwedge r' \in
    ( \mathcal{D}_{1} \bowtie \mathcal{D}_{4}) \nltimes \mathcal{D}_{3}
  \right)
\\ \iff \enspace &
r_1 r_2 \in \mathcal{U}
\blwedge
\blneg (\exists r_1' \in \Tupl{R_1}) (\exists r_4 \in \Tupl{R_4})
\\ &
\left(
  r_1' \between r_4 \blwedge r_1' = r_1 \blwedge r_1 r_2 \between r_4
  \blwedge
  r_1' r_4 \in
    ( \mathcal{D}_{1} \bowtie \mathcal{D}_{4}) \nltimes \mathcal{D}_{3}
\right)
\\ \iff \enspace &
r_1 r_2 \in \mathcal{U}
\blwedge
\blneg (\exists r_4 \in \Tupl{R_4})
\left(
  r_1 r_2 \between r_4 \blwedge r_1 r_4 \in
    ( \mathcal{D}_{1} \bowtie \mathcal{D}_{4}) \nltimes \mathcal{D}_{3}
\right)
\\ \iff \enspace &
r_1 r_2 \in \mathcal{U}
\blwedge
\blneg (\exists r_4 \in \Tupl{R_4})
\\ &
\left(
  r_1 r_2 \between r_4 \blwedge r_1 r_4 \in
    \mathcal{D}_{1} \bowtie \mathcal{D}_{4}
    \blwedge
    \blneg (\exists r_3 \in \Tupl{R_3})
    \left(
    r_1 r_4 \between r_3 \blwedge r_3 \in \mathcal{D}_{3}
    \right)
\right)
\end{align*}

Now, $r_1 r_4 \in \mathcal{D}_{1} \bowtie \mathcal{D}_{4}$ is equivalent to
$r_1 \in \mathcal{D}_{1} \blwedge r_4 \in \mathcal{D}_{4}$ provided that $r_1$ is
joinable with $r_4$, but this is ensured by $r_1 r_2 \between r_4$. Furthermore,
 $r_1 \in \mathcal{D}_{1}$ does not depend on the existence of $r_4$ and can be
 taken outside the scope of the quantifier. We get
\begin{align*}
& r_1 r_2 \in
  \mathcal{D}_1 \div^{\mathcal{D}_3,\mathcal{D}_4}_\mathrm{ddo} \mathcal{D}_2
\\ \iff \enspace &
r_1 r_2 \in \mathcal{U}
\blwedge
\blneg (\exists r_4 \in \Tupl{R_4})
\\ &
\left(
  r_1 r_2 \between r_4 \blwedge r_1 r_4 \in
    \mathcal{D}_{1} \bowtie \mathcal{D}_{4}
    \blwedge
    \blneg (\exists r_3 \in \Tupl{R_3})
    \left(
      r_1 r_4 \between r_3 \blwedge r_3 \in \mathcal{D}_{3}
    \right)
\right)
\\ \iff \enspace &
r_1 r_2 \in \mathcal{U}
\blwedge
\blneg (r_1 \in \mathcal{D}_{1} \blwedge (\exists r_4 \in \Tupl{R_4})
\\ &
\left(
  r_1 r_2 \between r_4 \blwedge r_4 \in \mathcal{D}_{4}
  \blwedge
  \blneg (\exists r_3 \in \Tupl{R_3})
  \left( r_1 r_4 \between r_3 \blwedge r_3 \in \mathcal{D}_{3} \right)
\right)
\\ \iff \enspace &
\overbrace{(r_1 r_2 \in \mathcal{U}
  \blwedge
  \blneg r_1 \in \mathcal{D}_{1})}^{\text{always false}
}
\blvee
(r_1 r_2 \in \mathcal{U}
\blwedge
\blneg (\exists r_4 \in \Tupl{R_4})
\\ &
\left(
  r_1 r_2 \between r_4 \blwedge r_4 \in \mathcal{D}_{4}
    \blwedge
    \blneg (\exists r_3 \in \Tupl{R_3})
    \left(
      r_1 r_4 \between r_3 \blwedge r_3 \in \mathcal{D}_{3})
    \right)
\right)
\\ \iff \enspace &
r_1 r_2 \in \mathcal{U}
\blwedge
(\forall r_4 \in \Tupl{R_4})
\\ &
\left(
  \blneg(r_1 r_2 \between r_4 \blwedge r_4 \in \mathcal{D}_{4})
  \blvee
  \blneg\blneg (\exists r_3 \in \Tupl{R_3})
  \left(
    r_1 r_4 \between r_3 \blwedge r_3 \in \mathcal{D}_{3}
  \right)
\right)
\\ \iff \enspace &
r_1 r_2 \in \mathcal{U}
\blwedge
(\forall r_4 \in \Tupl{R_4})
\\ &
\left(
  (r_1 r_2 \between r_4 \blwedge r_4 \in \mathcal{D}_{4})
  \blRightarrow
  (\exists r_3 \in \Tupl{R_3})
  \left(
    r_1 r_4 \between r_3 \blwedge r_3 \in \mathcal{D}_{3}
  \right)
\right),
\end{align*}
which concludes the proof.
\end{proof}
\end{theorem}

Now, based on~\eqref{eqn:DarwenOriginal_set}, we may introduce a graded variant
$\div_\mathrm{gddo}$ of the Darwen's Divide as follows
\begin{align}
  \bigl(
  \mathcal{D}_1 &
  \div^{\mathcal{D}_3,\mathcal{D}_4}_\mathrm{gddo}
  \mathcal{D}_2
  \bigr)(r_1 r_2) = \notag \\ &=
  \mathcal{D}_1(r_1) \otimes \mathcal{D}_2(r_2) \otimes
  \bigwedge_{\substack{r_4 \in \Tupl{R_4} \\ r_1 r_2 \between r_4}}
  \!\!\!\!\!\!\!\!\!
  \Bigl( \mathcal{D}_{4}(r_4)
  \rightarrow
  \!\!\!
  \bigvee_{\substack{r_3 \in \Tupl{R_3} \\r_1 r_4 \between r_3}}
  \!\!\!\!\!\!\!\!\!
  \mathcal{D}_{3}(r_3) \Bigr) \notag \\
  &=
  (\join{}{\mathcal{D}_1}{\mathcal{D}_2})(r_1 r_2)
  \otimes
  \bigwedge_{\substack{r_4 \in \Tupl{R_4} \\ r_1 r_2 \between r_4}}
  \!\!\!\!\!\!\!\!\!
  \Bigl( \mathcal{D}_{4}(r_4)
  \rightarrow
  \!\!\!
  \bigvee_{\substack{r_3 \in \Tupl{R_3} \\r_1 r_4 \between r_3}}
  \!\!\!\!\!\!\!\!\!
  \mathcal{D}_{3}(r_3) \Bigr)
  \label{def:gradedDarwen}
\end{align}

The condition of joinability is not necessary and can be avoided. We can put
$R_{4\setminus 12} = R_4 \setminus (R_1 \cup R_2)$,
$R_{4\cap 12} = R_4 \cap (R_1 \cup R_2)$,
$R_{3\setminus 14} = R_3 \setminus (R_1 \cup R_4)$ and
$R_{3\cap 14} = R_3 \cap (R_1 \cup R_4)$. Obviously, it holds that
$R_{4\setminus 12} \cap R_{4\cap 12} = \emptyset$ and
$R_{4\setminus 12} \cup R_{4\cap 12} = R_4$. The same holds for
$R_{3\setminus 14}$ and $R_{3\cap 14}$.
Now, denote by $\TupCommon{1}{2}{4} = (r_1 r_2) (R_{4\cap 12})$ the projection of
 tuple $r_1 r_2$ onto $R_{4\cap 12}$ (i.e. onto common attributes of $R_4$ and
 $R_1 \cup R_2$. Considering $r_4' \in \Tupl{R_{4\setminus 12}}$ we get
$r_4 = \TupCommon{1}{2}{4} r_4'$.
Observe, that tuples $\TupCommon{1}{2}{4}$ and $r_4'$ are always joinable since
$R_{4\setminus 12} \cap R_{4\cap 12} = \emptyset$. We have expressed the tuple
$r_4$ without any need for joinability condition and we can remove the condition
 from the infimum operation.

We can now proceed to the condition in supremum. Note, that
\[
r_1 r_4
= r_1 \TupCommon{1}{2}{4} r_4'
= r_1 (r_1 r_2) (R_{4\cap 12}) r_4'
= r_1 (r_2) (R_{4\cap 2}) r_4'
= r_1 \TupCommon{}{2}{4} r_4'
\]

Again, by $\TupCommon{1}{4}{3} = (r_1 r_4) (R_{3\cap 14})
= (r_1 \TupCommon{}{2}{4} r_4')(R_{3\cap 14}) $ we denote the projection of the
tuple in question onto $R_{3\cap 14}$. For $r_3' \in \Tupl{R_{3\setminus 14}}$
we get $r_3 = \TupCommon{1}{4}{3} r_3'$.
Using similar argument, $\TupCommon{1}{4}{3}$ and $r_3'$ are always joinable.

Putting both observations together we finally get
\begin{align}
  \bigl(
  \mathcal{D}_1 &
  \div^{\mathcal{D}_3,\mathcal{D}_4}_\mathrm{gddo}
  \mathcal{D}_2
  \bigr)(r_1 r_2) = \notag \\ &=
  (\join{}{\mathcal{D}_1}{\mathcal{D}_2})(r_1 r_2) \otimes
  \bigwedge_{\substack{r_4 \in \Tupl{R_4} \\ r_1 r_2 \between r_4}}
  \!\!\!\!\!\!\!\!\!
  \Bigl( \mathcal{D}_{4}(r_4)
  \rightarrow
  \!\!\!
  \bigvee_{\substack{r_3 \in \Tupl{R_3} \\r_1 r_4 \between r_3}}
  \!\!\!\!\!\!\!\!\!
  \mathcal{D}_{3}(r_3) \Bigr) \notag \\
  &=
  (\join{}{\mathcal{D}_1}{\mathcal{D}_2})(r_1 r_2) \otimes
  \bigwedge_{r_4' \in \Tupl{R_{4\setminus 12}}}
  \!\!\!\!\!\!\!\!\!
  \Bigl( \mathcal{D}_{4}(\TupCommon{1}{2}{4} r_4')
  \rightarrow
  \!\!\!
  \bigvee_{r_3' \in \Tupl{R_{3\setminus 14}}}
  \!\!\!\!\!\!\!\!\!
  \mathcal{D}_{3}(\TupCommon{1}{4}{3} r_3') \Bigr)
  \label{def:gradedDarwenNoCond}
  \\ &=
  (\join{}{\mathcal{D}_1}{\mathcal{D}_2})(r_1 r_2) \otimes
  \bigwedge_{r_4' \in \Tupl{R_{4\setminus 12}}}
  \!\!\!\!\!\!\!\!\!
  \Bigl( \mathcal{D}_{4}(\TupCommon{1}{2}{4} r_4')
  \rightarrow
  \pi_{R_{3\cap 14}} (\mathcal{D}_{3}) ( \TupCommon{1}{4}{3} )
  \Bigr)
  \label{def:gradedDarwenNoCondAlt}
\end{align}

Graded Date's Great and Small Divide can be easily expressed by the graded
version of Darwen's Divide in the following way.

\begin{theorem}
  For relations on schemes that conform to requirements for Great Divide,
  precisely for relations $\mathcal{D}_1$ on $R$, $\mathcal{D}_2$ on $T$,
  $\mathcal{D}_3$ on $RS$, and $\mathcal{D}_4$ on $ST$, we have
  \[ \mathcal{D}_1 \div^{\mathcal{D}_3,\mathcal{D}_4}_\mathrm{ggdo}
  \mathcal{D}_2 = \mathcal{D}_1 \div^{\mathcal{D}_3,\mathcal{D}_4}_\mathrm{gddo}
   \mathcal{D}_2. \]
\end{theorem}
\begin{proof}
  For relations $\mathcal{D}_1$ on $R$, $\mathcal{D}_2$ on $T$, $\mathcal{D}_3$
  on $RS$ and $\mathcal{D}_4$ on $ST$, we have
  \begin{align*}
  \bigl(
  \mathcal{D}_1 &
  \div^{\mathcal{D}_3,\mathcal{D}_4}_\mathrm{gddo}
  \mathcal{D}_2
  \bigr)(r t) =
  \\ &=
  (\join{}{\mathcal{D}_1}{\mathcal{D}_2})(r_1 r_2) \otimes
  \bigwedge_{r_4' \in \Tupl{ST \setminus (R \cup T)}}
  \!\!\!\!\!\!\!\!\!
  \Bigl( \mathcal{D}_{4}(\TupCommon{1}{2}{4} r_4')
  \rightarrow
  \pi_{RS \cap (R \cup ST) } (\mathcal{D}_{3}) ( \TupCommon{1}{4}{3} )
  \Bigr)
  \\ &=
  (\join{}{\mathcal{D}_1}{\mathcal{D}_2})(r t)
  \otimes
  \bigwedge_{\substack{s \in \Tupl{S}}}
  \!\!\!\!\!\!
  \Bigl( \mathcal{D}_{4}(st) \rightarrow \mathcal{D}_{3}(rs) \Bigr)
  \\ &= \bigl(
  \mathcal{D}_1
  \div^{\mathcal{D}_3,\mathcal{D}_4}_\mathrm{ggdo}
  \mathcal{D}_2
  \bigr)(r t).
\end{align*}
\end{proof}

\begin{corollary}
For relations on schemes that conform to requirements for Small Divide,
precisely for relations $\mathcal{D}_1$ on $R$, $\mathcal{D}_2$ on $S$ and
$\mathcal{D}_3$ on $RS$, we have $\mathcal{D}_1
\div_\mathrm{gsdo}^{\mathcal{D}_3} \mathcal{D}_2 =
\mathcal{D}_1 \div_\mathrm{gddo}^{\mathcal{D}_3,\mathcal{D}_2} \Dee{1}$.
\qed
\end{corollary}

%%%%%%%%%%%%%%%%%%%%%%%%%%%%%%%%%%%%%%%%%%%%%%%%%%%%%%%%%%%%%%%%%%%%%%%%%%%%%%%%
%%%%%   PTC
%%%%%%%%%%%%%%%%%%%%%%%%%%%%%%%%%%%%%%%%%%%%%%%%%%%%%%%%%%%%%%%%%%%%%%%%%%%%%%%%
\section{Pseudo Tuple Relational Calculus}\label{sec:ptc}
In this section, we present a query language we use in this paper for easier
reasoning about the relational algebra operations. The Pseudo Tuple Calculus
(shortly, PTC) is similar to the ordinary tuple calculus, however, 
it provides more convenient way to reason about relational algebra expressions 
in the presence of scores. In the next section we use the PTC to show mutual
relationships among the division operations.

%%%%%%%%%%%%%%%%%%%%%%%%%%%%%%%%%%%%%%%%%%%%%%%%%%%%%%%%%%%%%%%%%%%%%%%%%%%%%%%%
\subsection{PTC-expressions and their evaluation}
Every PTC-expression $\mathcal{T}(\tv{r}_1, \ldots, \tv{r}_n)$ of Pseudo Tuple
Calculus is associated with a finite set of free tuple variables
$\tv{r}_1, \ldots, \tv{r}_n$ that appear in the PTC-expression.
For each tuple variable $\tv{r}_i$ we consider its relation scheme $R_i$. We
assume that tuple variables with the same name have the same relation scheme.
The relation scheme $R_\mathcal{T}$ of PTC-expression
$\mathcal{T}(\tv{r}_1, \ldots, \tv{r}_n)$ is given by the union of relation
schemes of the tuple variables $R_\mathcal{T} = \bigcup_{i = 1}^{n} R_i$.

Since we do not utilize any disjunctive operations in this paper we define here
only a fragment of the Pseudo Tuple Calculus without the
corresponding disjunctive expressions. For the same reason we omit the
treatment of restrictions as well.

%%%%%%%%%%%%%%%%%%%%%%%%%%%%%%%%%%%%%%%%%%%%%%%%%%%%%%%%%%%%%%%%%%%%%%%%%%%%%%%%
\subsubsection{Syntax of PTC-expressions}
The PTC-expressions are defined inductively as follows.
\begin{enumerate}
\item if $E$ is a relational algebra expression (shortly,
RA-expression) on relation scheme $R$ and $\tv{r}_1, \ldots, \tv{r}_n$ are tuple
variables on $R_1, \ldots, R_n$ such that $R = \bigcup_{i = 1}^{n} R_i$, then
$E(\tv{r}_1, \ldots, \tv{r}_n)$ is an (atomic) PTC-expression on relation scheme
$R$.

In order to keep our notation simple, we abbreviate finite sets of tuple
variables as $\mathbf{r} = \{ \tv{r}_1, \ldots, \tv{r}_n \}$ and their
corresponding relation schemes as $R_\mathbf{r} = \bigcup_{i = 1}^{n} R_i$.
In the simplified notation, the (atomic) PTC-expression
$E(\tv{r}_1, \ldots, \tv{r}_n)$ becomes $E(\mathbf{r})$.

\item if $\mathcal{T}_1(\mathbf{r}_1)$ and $\mathcal{T}_2(\mathbf{r}_2)$ are
PTC-expressions on $R_{\mathbf{r}_1}$ and $R_{\mathbf{r}_2}$ respectively, then
$(\mathcal{T}_1(\mathbf{r}_1)
\circop
\mathcal{T}_2(\mathbf{r}_2))(\mathbf{r}_1 \cup \mathbf{r}_2)$
is PTC-expression on $R_{\mathbf{r}_1} \cup R_{\mathbf{r}_2}$, where $\circop$
is one of the following symbols $\otimes, \wedge, \rightarrow$. Note that
$\mathbf{r}_1 \cup \mathbf{r}_2$ is well-defined since we assume that tuple
variables with the same name have the same relation scheme.

To simplify notation, we do not have to explicitly mention the set
$\mathbf{r}_1 \cup \mathbf{r}_2$ since it can be easily deduced from the form
of the subexpressions. Thus, the above mentioned PTC-expression becomes
$\mathcal{T}_1(\mathbf{r}_1) \circop \mathcal{T}_2(\mathbf{r}_2)$. In more
complex expressions we utilize outer parentheses to avoid ambiguity in the usual
 way.

\item if $\mathcal{T}(\mathbf{r})$ is PTC-expression on $R_{\mathbf{r}}$ then
$(\nabla\mathcal{T}(\mathbf{r}))(\mathbf{r})$ and
$(\Delta\mathcal{T}(\mathbf{r}))(\mathbf{r})$ are PTC-expressions on
$R_{\mathbf{r}}$.
In simplified notation we have $\nabla\mathcal{T}(\mathbf{r})$ and
$\Delta\mathcal{T}(\mathbf{r})$.

\item if $\mathcal{T}(\mathbf{r}_1 \cup \mathbf{r}_2)$ is PTC-expression on
$R_{\mathbf{r}_1} \cup R_{\mathbf{r}_2}$ such
that $R_{\mathbf{r}_1} \cap R_{\mathbf{r}_2} = \emptyset$ then
$\left( \bigvee_{\!\mathbf{r}_1}
 \! \mathcal{T}(\mathbf{r}_1 \cup \mathbf{r}_2) \right) (\mathbf{r_2})$ and
$\left( \bigwedge_{\mathbf{r}_1}
 \! \mathcal{T}(\mathbf{r}_1 \cup \mathbf{r}_2) \right) (\mathbf{r_2})$ are
PTC-expressions on~$R_{\mathbf{r}_2}$.

For aesthetic reasons we will denote the set $\mathbf{r}_1 \cup \mathbf{r}_2$
by $\mathbf{r}_1, \mathbf{r}_2$. In the simplified notation we get
$\ptcsup$ and $\ptcinf$.
\end{enumerate}

%%%%%%%%%%%%%%%%%%%%%%%%%%%%%%%%%%%%%%%%%%%%%%%%%%%%%%%%%%%%%%%%%%%%%%%%%%%%%%%%
\subsubsection{Semantics of PTC-expressions}
The evaluation of PTC-expressions is based on the notion of a database
instance~$\DIN$. Loosely speaking, a database instance assigns appropriate
relations to relation symbols from a database scheme---database instance can be
seen as a snapshot of all base relations that we have in some database.
Relations in database naturally change in time, however the database instance is
 fixed as it reflects the state of the database in a given point of time.
We tacitly assume that the database scheme is clear from the context.

Furthermore, we utilize the notion of extended active
domains. First, we define the active domain $adom(y, \mathcal{D})$ for the given
 attribute $y$ and relation $\mathcal{D}$ as a projection of $\mathcal{D}$ onto
${y}$ where all tuples with non-zero scores have their score set to one.
 We denote by $\mathcal{D}_i^y$ ($i \in I$) relations from the
given database instance $\DIN$ whose relation schemes contain the attribute
$y$ ($\{y\} \subseteq R_i$).
The extended active domain $eadom^{\DIN}(y)$ for the given database instance
$\DIN$ and attribute $y$ is defined as
\[ eadom^{\DIN}(y) = \textstyle\bigcup_{i \in I} adom(y, \mathcal{D}_i^y). \]
For the entire relation scheme $R = \{y_1, \ldots, y_n\}$ we define the extended
active domain as
\[ eadom^{\DIN}_R = eadom^{\DIN}(y_1) \bowtie \cdots \bowtie eadom^{\DIN}(y_n). \]
It is easy to see that the $eadom^{\DIN}_R$ contains every tuple on relation
scheme $R$ that can be built from all values of respective domains that are
available in the database instance in question. The extended active domain can
be seen as a finite universe of tuples for the given database instance and
relation scheme if we do not allow introduction of new domain values (by
singleton relations).

\begin{remark}
As an aside, let us mention that it is easy to modify the definition of extended
 active domain to incorporate new values introduced by singleton relations.
Since RA-expressions are finite, the number of new values is finite as well.
Before obtaining extended active domain $eadom^{\DIN}_R$ by joining all extended
 active domains $eadom(y_i)$ for attributes
$y_i \in R$ $(i \in \{ 1, \ldots, n \})$
it suffices to unify each $eadom(y_i)$ with a (finite) set of new values
whose domain coincides with the domain of attribute $y_i$.
\end{remark}

Now, we define the evaluation of PTC-expressions in database instances.
Suppose we have the database instance $\DIN$ and a PTC-expression
$\mathcal{T}(\mathbf{r})$, where
$\mathbf{r} = \{ \blmathbb{r}_1, \ldots, \blmathbb{r}_n \}$
such that each tuple variable $\blmathbb{r}_i$ is on relation scheme $R_i$.
By evaluating $\mathcal{T}(\mathbf{r})$ in $\DIN$ we obtain a relation
$\mathcal{T}^{\DIN}$ on relation scheme $R = \bigcup_{i = 1}^{n} R_i$.
For any tuple $r \not\in eadom^{\DIN}_R$ we put $\mathcal{T}^{\DIN}(r) = 0$.
In other words, the relation $\mathcal{T}^{\DIN}$ may contain only tuples from
$eadom^{\DIN}_R$. For each tuple $r \in eadom^{\DIN}_R$ we define its score in
the relation $\mathcal{T}^{\DIN}$ as follows.

Any tuple $r \in eadom^{\DIN}_R$ induces a valuation of the tuple variables
$\blmathbb{r}_1, \ldots, \blmathbb{r}_n$ from the PTC-expression. The valuation
assigns each variable $\blmathbb{r}_i$ the projection of tuple $r$ onto the
relation scheme $R_i$ of the variable in question, symbolically
$\| \blmathbb{r}_i \|_{r} = r(R_i)$.
We denote the join of valuated tuple variables
$\| \blmathbb{r}_1 \|_{r} \cdots \| \blmathbb{r}_n \|_{r}$ as
$ \| \mathbf{r} \|_{r}$. It is easy to see that $ \| \mathbf{r} \|_{r} = r$.
In general, for a set of tuple variables $\mathbf{r'}$ such that
$\mathbf{r'} \subseteq \mathbf{r}$ with relation scheme
$R_{\mathbf{r'}} \subseteq R$ it holds that
$\| \mathbf{r'} \|_{r} = r(R_\mathbf{r'}) $.
We define the score $\mathcal{T}^{\DIN}( \| \mathbf{r} \|_{r} )$ of tuple
$\| \mathbf{r} \|_{r}$ in the relation $\mathcal{T}^{\DIN}$ as follows.
According to the form of PTC-expression we distinguish the following cases
\begin{enumerate}
  \item if $\mathcal{T}(\mathbf{r})$ is $E(\mathbf{r})$,
  we first evaluate the RA-expression $E$ in the database instance ${\DIN}$
  according to RA-expression evaluation rules (\cite{BeVy:TODS}) and denote
  the resulting relation as $E^{\DIN}$, then we set
  $\mathcal{T}^{\DIN}( \| \mathbf{r} \|_{r} ) =
    E^{\DIN}( \| \mathbf{r} \|_{r} )$,
  \item if $\mathcal{T}(\mathbf{r})$ is
  $\mathcal{T}_1(\mathbf{r}_1) \circop \mathcal{T}_2(\mathbf{r}_2)$,
  where $\circop$ is on of the following symbols $\otimes, \wedge,
  \rightarrow$, and $\mathbf{r} = \mathbf{r}_1 \cup \mathbf{r}_2$,
  first we get the scores $\mathcal{T}^{\DIN}_1( \| \mathbf{r}_1 \|_{r} )$ and
   $\mathcal{T}^{\DIN}_2( \| \mathbf{r}_2 \|_{r} )$
  with the valuation induced by $r$. Then we set
  $\mathcal{T}^{\DIN}( \| \mathbf{r} \|_{r} ) =
    \mathcal{T}^{\DIN}_1(\| \mathbf{r}_1 \|_{r})
    \mathop{\circ}
    \mathcal{T}^{\DIN}_2(\| \mathbf{r}_2 \|_{r})$,
  where $\circop$ is one of the following operations
  $\otimes, \wedge, \rightarrow$.
  \item if $\mathcal{T}(\mathbf{r})$ is $\nabla\mathcal{T'}(\mathbf{r})$ or
  $\Delta\mathcal{T'}(\mathbf{r})$ we get the score
  $\mathcal{T'}^{\DIN}( \| \mathbf{r} \|_{r} )$.

  If $\mathcal{T}(\mathbf{r})$ is $\nabla\mathcal{T'}(\mathbf{r})$ we set
  \[
    \mathcal{T}^{\DIN}( \| \mathbf{r} \|_{r} ) =
      \begin{cases}
        1 & \text{if } \mathcal{T'}^{\DIN}(\| \mathbf{r} \|_r) > 0, \\
        0 & \text{otherwise.}
      \end{cases}
  \]
  If $\mathcal{T}(\mathbf{r})$ is $\Delta\mathcal{T'}(\mathbf{r})$ we set
  \[
    \mathcal{T}^{\DIN}( \| \mathbf{r} \|_{r} ) =
      \begin{cases}
        1 & \text{if } \mathcal{T'}^{\DIN}(\| \mathbf{r} \|_r) = 1, \\
        0 & \text{otherwise.}
      \end{cases}
  \]
\item if $\mathcal{T}(\mathbf{r})$ is $\ptcsupprime$ or $\ptcinfprime$ where
$\mathbf{r} = \mathbf{r_2}$, first we get the scores
$\mathcal{T'}^{\DIN}( \| \mathbf{r_1} \cup \mathbf{r_2} \|_{r r'} )$ with the
valuation induced by the join of tuples $r$ and $r'$ for every
$r' \in eadom^{\DIN}_{R_{\mathbf{r}_1}}$. Note that the tuples $r$ and $r'$ are
always joinable as the relation schemes $R_{\mathbf{r}_1}$ and
$R_{\mathbf{r}_2} = R$ are disjoint (from the definition of PTC-expression).
Since $eadom^{\DIN}_{R_{\mathbf{r}_1}}$ is finite, we obtain a finite set of
scores $\{ \mathcal{T'}^{\DIN}( \| \mathbf{r_1} \cup \mathbf{r_2} \|_{r r'} )
\mid
r' \in eadom^{\DIN}_{R_{\mathbf{r}_1}} \}$.

If $\mathcal{T}(\mathbf{r})$ is $\ptcsupprime$ we set
\[
  \mathcal{T}^{\DIN}( \| \mathbf{r} \|_{r} ) =
  \textstyle\bigvee
  \{ \mathcal{T'}^{\DIN}( \| \mathbf{r_1} \cup \mathbf{r_2} \|_{r r'} )
  \mid
  r' \in eadom^{\DIN}_{R_{\mathbf{r}_1}} \}.
\]
If $\mathcal{T}(\mathbf{r})$ is $\ptcinfprime$ we set
\[
  \mathcal{T}^{\DIN}( \| \mathbf{r} \|_{r} ) =
  \textstyle\bigwedge
  \{ \mathcal{T'}^{\DIN}( \| \mathbf{r_1} \cup \mathbf{r_2} \|_{r r'} )
  \mid
  r' \in eadom^{\DIN}_{R_{\mathbf{r}_1}} \}.
\]
\end{enumerate}

%%%%%%%%%%%%%%%%%%%%%%%%%%%%%%%%%%%%%%%%%%%%%%%%%%%%%%%%%%%%%%%%%%%%%%%%%%%%%%%%
\subsubsection{Splitting principle}
Consider a PTC-expression $\mathcal{T}(\ldots,\blmathbb{r},\ldots)$ such that the
tuple variable $\blmathbb{r}$ is on relation scheme $R$. If we replace the tuple
 variable $\blmathbb{r}$ with two (or more) fresh tuple variables
$\blmathbb{r}_1, \blmathbb{r}_2$ on $R_1$ and $R_2$ such that
$R_1 \cup R_2 = R$, we obtain a PTC-expression
$\mathcal{T'}(\ldots,\blmathbb{r}_1, \blmathbb{r}_2,\ldots)$ that differs only in
the set of free variables. Despite being different on the syntactic level it is
straightforward to see that for any database instance $\DIN$ we have
$\mathcal{T}^{\DIN} = \mathcal{T'}^{\DIN}$.

From the semantic point of view, we are free to ``split'' free tuple variables
and ``join'' them back without changing the meaning of the PTC-expression.
We call this ``the splitting principle''.

%%%%%%%%%%%%%%%%%%%%%%%%%%%%%%%%%%%%%%%%%%%%%%%%%%%%%%%%%%%%%%%%%%%%%%%%%%%%%%%%
\subsection{Equivalence of PTC and Relational Algebra}
In this section we show that the Pseudo Tuple Calculus and Relational Algebra
are equivalent. First, observe that if we evaluate any RA-expression $E$ on
relation scheme $R$ in a database instance $\DIN$,
the relation $E^{\DIN}$ may contain tuples from $eadom^{\DIN}_R$ only, since
$eadom^{\DIN}_R$ consists of all tuples that can possibly be built from the
values available in the database instance, i.\,e. we have
\begin{align}\label{PTC:RaScope}
E^{\DIN}(r) = 0 \text{ whenever } r \not\in eadom^{\DIN}_R.
\end{align}
Using this observation we can easily prove the following theorem.
\begin{theorem}\label{PTC:RaToPTC}
For any RA-expression $E$ on relation scheme $R$
there is a PTC-expression $\mathcal{T}(\mathbf{r})$ on $R$ such that
for any database instance $\DIN$ we have
$E^{\DIN}(r) = \mathcal{T}^{\DIN}(r)$ for all $r \in \Tupl{R}$.
\end{theorem}
\begin{proof}
Since any RA-expression is directly an (atomic) PTC-expression we can take
$E(\blmathbb{r})$ with a single tuple variable $ \blmathbb{r} $ on relation
scheme $R$ as the sought PTC-expression $\mathcal{T}(\mathbf{r})$.
From the definition of PTC-expression evaluation and the observation
\eqref{PTC:RaScope} we conclude that $E^{\DIN}(r) = \mathcal{T}^{\DIN}(r) $
holds for all $r \in \Tupl{R}$.
\end{proof}
It follows that the Pseudo Tuple Calculus is at least as powerful as the
Relational Algebra. Before proving the converse theorem we need one more
observation. Recall that the relation $eadom^{\DIN}_R$ plays an important role
in PTC-expression evaluation as it serves the purpose of an implicit range
(or universe) for evaluation. Since evaluation of RA-expressions is
unconstrained and takes all tuples in account we need to be able to construct a
RA-expression $\RAeadom{R}$ that will evaluate to $eadom^{\DIN}_R$ and will act
as an explicit range for evaluation of RA-expressions.

It is easy to see that the active domain $adom(y, \mathcal{D})$ for
the given attribute $y$ and relation $\mathcal{D}$ can be computed by evaluating
the RA-expression $\mathcal{A}_y(\mathbb{D}) = \pi_{\{y\}}(\nabla \mathbb{D})$
in database instance $\DIN$, where $\mathbb{D}$ is a relation symbol evaluated
to relation $\mathcal{D}$ by the database instance. For the extended active
domain $eadom^{\DIN}(y)$ for an attribute $y$ the RA-expression is
$\mathcal{E}_y = \bigcup_{i \in I} \mathcal{A}_y(\mathbb{D}_i^y)$,
where $\mathbb{D}_i^y$ are
%\bigcup_{1 \leq j \leq n} [y \colon \mathfrak{c}_i]$, where $\mathbb{D}_i^y$
%are relation symbols whose relation scheme contains attribute $y$ and
%$\mathfrak{c}_i$ are constant symbols
relation symbols whose relation scheme contains attribute $y$ and the database
instance $\DIN$ interprets each relation symbol $\mathbb{D}_i^y$
as relation $\mathcal{D}_i^y$.
% and each constant symbol $\mathfrak{c}_i$ as constant $\mathfrak{c}_i^{\DIN}$.
Finally, we get the extended active domain $eadom^{\DIN}_R$ for scheme
$R = \{ y_1, \ldots, y_n \}$
%and $\mathfrak{C} = \mathfrak{C}_1 \cup \dots \cup \mathfrak{C}_n$
by evaluating $\RAeadom{R} =
\mathcal{E}_{y_1} \bowtie \cdots \bowtie \mathcal{E}_{y_n}$
in database instance $\DIN$. In other words we have
$eadom^{\DIN}_R = \RAeadomEV{R}$.

\begin{remark}
As an aside, if we use the modified definition of extended active domain that
 allows introduction of new values by singleton relations, we need to modify the
 previous definition of $\mathcal{E}_y$ to reflect the extended meaning of
$eadom^{\DIN}_R$. The definition becomes
$\mathcal{E}_y =
  \bigcup_{i \in I} \left( \mathcal{A}_y(\mathbb{D}_i^y) \right) \cup
  \bigcup_{j=1}^{n} [y {\colon} \! \mathfrak{c}_j]$, where
$\mathbb{D}_i^y$ are relation symbols whose relation scheme contains attribute
$y$ and $\mathfrak{c}_j$ are symbols denoting new values from the domain
of attribute $y$ such that the database instance $\DIN$ interprets each relation
 symbol $\mathbb{D}_i^y$ as relation $\mathcal{D}_i^y$ and each symbol
$\mathfrak{c}_j$ as the new value $\mathfrak{c}_j^{\DIN}$ from the respective
domain.
\end{remark}

\begin{theorem}\label{PTC:PTCtoRA}
For any PTC-expression $\mathcal{T}(\mathbf{r})$ with
 $\mathbf{r} = \{ \blmathbb{r}_1, \ldots, \blmathbb{r}_n \}$, where the
 tuple variables $\blmathbb{r}_i$ are on relation schemes $R_i$, there is a
 RA-expression $F$ on relation scheme $R = \bigcup_{i=1}^{n} R_i$ such that
 for any database instance $\DIN$ we have
$F^{\DIN}(r) = \mathcal{T}^{\DIN}(r)$ for all $r \in \Tupl{R}$.
\end{theorem}
\begin{proof}
The theorem is proved by induction on the complexity of the PTC-expression.
In each step, we show the RA-expression $F$ that forms the counterpart to the
PTC-expression $\mathcal{T}(\mathbf{r})$ in question. Furthermore, we show that
the results of evaluating both RA- and PTC-expression coincide, i.\,e. the
relations $F^{\DIN}$ and $\mathcal{T}^{\DIN}$ have the same relation scheme and
contain the same tuples. Recall that for a set of tuple variables
$\mathbf{r'}$ on the relation scheme $R_{\mathbf{r'}}$ and a tuple
$r \in \Tupl{R}$ such that $R_{\mathbf{r'}} \subseteq R$ we have
$ \| \mathbf{r'} \|_{r} = r(R_{\mathbf{r'}}) $.

Let us have a PTC-expression $\mathcal{T}(\mathbf{r})$, where
$\mathbf{r} = \{ \blmathbb{r}_1, \ldots, \blmathbb{r}_n \}$
such that each tuple variable $\blmathbb{r}_i$ is on relation scheme $R_i$.
The relation scheme of the relation $\mathcal{T}^{\DIN}$ is
$R = \bigcup_{i=1}^{n} R_i$. We obtain the equivalent RA-expression $F$ as
follows.
\begin{enumerate}
  \item If $\mathcal{T}(\mathbf{r})$ is $E(\mathbf{r})$, the sought
  RA-expression $F$ is $E$.

  Since the relation scheme of $F$ is $R$, the relations $\mathcal{T}^{\DIN}$
  and $F^{\DIN}$ have the same relation scheme.
  For any tuple $r \in eadom^{\DIN}_R$ we have
  \[ \mathcal{T}^{\DIN}(r)
    = \mathcal{T}^{\DIN}( \| \mathbf{r} \|_r )
    = E^{\DIN}(\| \mathbf{r} \|_r)
    = F^{\DIN}(r) \] from the definition of PTC-expression evaluation.

  From \eqref{PTC:RaScope} it follows that for all tuples
  $r \not\in eadom^{\DIN}_R$ we have $F^{\DIN}(r) = 0$. Together, we have
  $\mathcal{T}^{\DIN}(r) = F^{\DIN}(r)$ for all tuples $r \in \Tupl{R}$.
  \item If $\mathcal{T}(\mathbf{r})$ is
  $\mathcal{T}_1(\mathbf{r}_1) \circop \mathcal{T}_2(\mathbf{r}_2)$,
  where $\circ$ is on of the following symbols $\otimes, \wedge,
  \rightarrow$, $\mathbf{r} = \mathbf{r}_1 \cup \mathbf{r}_2$ and
  $R = R_{\mathbf{r}_1} \cup R_{\mathbf{r}_2}$,
  then from the induction hypothesis we have RA-expressions $E_1$ on relation
  scheme $R_{\mathbf{r}_1}$ and $E_2$ on relation scheme $R_{\mathbf{r}_2}$
  corresponding to PTC-subexpressions $\mathcal{T}_1(\mathbf{r}_1)$ and
  $\mathcal{T}_2(\mathbf{r}_2)$, respectively, such that
  $\mathcal{T}_1^{\DIN}(r_1) = E^{\DIN}_1(r_1)$ and
  $\mathcal{T}_2^{\DIN}(r_2) = E^{\DIN}_2(r_2)$ for all
  $r_1 \in \Tupl{R_1}$ and $r_2 \in \Tupl{R_2}$.

  According to the symbol $\circ$ we distinguish three cases:
  \begin{enumerate}
    \item If $\circ$ is $\otimes$, then we put $F = E_1 \bowtie E_2$.

    The relation scheme of $F$ is $R_{\mathbf{r}_1} \cup R_{\mathbf{r}_2}$
    as required. We have
    \begin{align*}
      \mathcal{T}^{\DIN}(r) & = \mathcal{T}^{\DIN}(\| \mathbf{r} \|_r) \\
      & = \mathcal{T}^{\DIN}_1(\| \mathbf{r}_1 \|_r)
        \otimes
        \mathcal{T}^{\DIN}_2(\| \mathbf{r}_2 \|_r)  \\
      & = \mathcal{T}^{\DIN}_1(r (R_{\mathbf{r}_1}) )
        \otimes
        \mathcal{T}^{\DIN}_2(r (R_{\mathbf{r}_2})) \\
      & = E^{\DIN}_1(r (R_{\mathbf{r}_1}))
        \otimes
        E^{\DIN}_2(r(R_{\mathbf{r}_2})) \\
      & = (E_1 \bowtie E_2)^{\DIN}
          (r (R_{\mathbf{r}_1}) r (R_{\mathbf{r}_2})) \\
      & = F^{\DIN}(r)
    \end{align*}
    for all tuples $r \in eadom^{\DIN}_R$.

    For any tuple $r \not\in eadom^{\DIN}_R$ either or both of
    $r (R_{\mathbf{r}_1}) \not\in eadom^{\DIN}_{R_{\mathbf{r}_1}}$ and
    $r (R_{\mathbf{r}_2}) \not\in eadom^{\DIN}_{R_{\mathbf{r}_2}}$ must
    hold, otherwise we would arrive at contradiction.
    Without loss of generality let us assume that
    $r (R_{\mathbf{r}_1}) \not\in eadom^{\DIN}_{R_{\mathbf{r}_1}}$.
    Then we have $\mathcal{T}^{\DIN}_1(r (R_{\mathbf{r}_1}) ) = 0$ and from
    the induction hypothesis we also have
    $E^{\DIN}_1(r (R_{\mathbf{r}_1})) = 0$.
    From the properties of $\otimes$ we conclude that $F^{\DIN}(r) = 0$ for
    $r \not\in eadom^{\DIN}_R$.
    \item If $\circ$ is $\wedge$, then we put
    $F = (E_1 \bowtie \RAeadom{R_{\mathbf{r}_2}})
       \cap
       (E_2 \bowtie \RAeadom{R_{\mathbf{r}_1}})$.

    Since both $E_1 \bowtie \RAeadom{R_{\mathbf{r}_2}}$ and
    $E_2 \bowtie \RAeadom{R_{\mathbf{r}_1}}$ are on relation scheme
    $R_{\mathbf{r}_1} \cup R_{\mathbf{r}_2}$, the RA-expression $F$ is
    well-defined and its relation scheme matches the relation scheme of the
    PTC-expression.

    Now, observe that for any tuple $r \in eadom^{\DIN}_R$ the following
    holds
    \[
      E_1^{\DIN}(r (R_{\mathbf{r}_1}) )
      = E_1^{\DIN}(r (R_{\mathbf{r}_1}) )
        \otimes
        \underbrace{eadom^{\DIN}_{R_{\mathbf{r}_2}}
          (r (R_{\mathbf{r}_2}) )}_{= 1}
      = (E_1 \bowtie \RAeadom{R_{\mathbf{r}_2}})^{\DIN}(r)
    \]
    Dually, it holds for $E^{\DIN}_2$ as well. To put the in words, we can
    ``extend'' the relation scheme of some relation without changing the
    scores of tuples in this relation.  Hence, we have
    \begin{align*}
      \mathcal{T}^{\DIN}(r) & =  \mathcal{T}^{\DIN}(\| \mathbf{r} \|_r) \\
      & = \mathcal{T}^{\DIN}_1(\| \mathbf{r}_1 \|_r)
        \wedge
        \mathcal{T}^{\DIN}_2(\| \mathbf{r}_2 \|_r)  \\
      & = \mathcal{T}^{\DIN}_1(r (R_{\mathbf{r}_1}) )
        \wedge
        \mathcal{T}^{\DIN}_2(r (R_{\mathbf{r}_2})) \\
      & = E^{\DIN}_1(r (R_{\mathbf{r}_1}))
        \wedge
        E^{\DIN}_2(r(R_{\mathbf{r}_2})) \\
      & = (E_1 \bowtie \RAeadom{R_{\mathbf{r}_2}})^{\DIN}(r)
        \wedge
        (E_2 \bowtie \RAeadom{R_{\mathbf{r}_1}})^{\DIN}(r) \\
      & = \left((E_1 \bowtie \RAeadom{R_{\mathbf{r}_2}})
        \cap
        (E_2 \bowtie \RAeadom{R_{\mathbf{r}_1}})\right)^{\DIN}(r) \\
      & = F^{\DIN}(r)
    \end{align*}
    for all tuples $r \in eadom^{\DIN}_R$.

    For any tuple $r \not\in eadom^{\DIN}_R$, use the same argument as for
    the case with $\otimes$ concluding that $F^{\DIN}(r) = 0$ for
    $r \not\in eadom^{\DIN}_R$.
    \item If $\circ$ is $\rightarrow$, then we put
    $F = (E_1 \bowtie \RAeadom{R_{\mathbf{r}_2}})
          \rightarrowtriangle^{\RAeadom{R}}
          (E_2 \bowtie \RAeadom{R_{\mathbf{r}_1}})$.

    Since all $E_1 \bowtie
    \RAeadom{R_{\mathbf{r}_2}}$, $E_2 \bowtie \RAeadom{R_{\mathbf{r}_1}}$, and
    $\RAeadom{R}$ are on relation scheme
    $R_{\mathbf{r}_1} \cup R_{\mathbf{r}_2}$, the RA-expression $F$ is
    well-defined and its relation scheme matches the relation scheme of the
    PTC-expression.

    Observe that since $\mathcal{T}^{\DIN}(r) > 0$ only for
    tuples $r \in eadom^{\DIN}_R$ and
    $ \RAeadomEV{R}(r) = eadom^{\DIN}_R(r) = 1$ for any tuple
    $r \in eadom^{\DIN}_R$, it holds that
    $\mathcal{T}^{\DIN}(r) =
      \RAeadomEV{R}(r)
      \otimes
      \mathcal{T}^{\DIN}(r).$
    Using previous observations we have
    \begin{align*}
    \mathcal{T}^{\DIN}(r)
      & = \RAeadomEV{R}(r)
        \otimes
        \mathcal{T}^{\DIN}( \| \mathbf{r} \|_r ) \\
      & = \RAeadomEV{R}(r)
        \otimes
        \left(\mathcal{T}^{\DIN}_1(\| \mathbf{r}_1 \|_r)
          \rightarrow
          \mathcal{T}^{\DIN}_2(\| \mathbf{r}_2 \|_r) \right)  \\
      & = \RAeadomEV{R}(r)
        \otimes
        \left(\mathcal{T}^{\DIN}_1(r (R_{\mathbf{r}_1}) )
          \rightarrow
          \mathcal{T}^{\DIN}_2(r (R_{\mathbf{r}_2})) \right) \\
      & = \RAeadomEV{R}(r)
        \otimes
        \left(E^{\DIN}_1(r (R_{\mathbf{r}_1}))
          \rightarrow E^{\DIN}_2(r(R_{\mathbf{r}_2})) \right) \\
      & = \RAeadomEV{R}(r)
        \otimes
        \left((E_1 \bowtie \RAeadom{R_{\mathbf{r}_2}})^{\DIN}(r)
          \rightarrow
          (E_2 \bowtie \RAeadom{R_{\mathbf{r}_1}})^{\DIN}(r) \right) \\
      & = \left( (E_1 \bowtie \RAeadom{R_{\mathbf{r}_2}})
        \rightarrowtriangle^{\RAeadom{R}}
        (E_2 \bowtie \RAeadom{R_{\mathbf{r}_1}}) \right)^{\DIN}(r) \\
      & = F^{\DIN}(r)
    \end{align*}
  \end{enumerate}
  for all tuples $r \in eadom^{\DIN}_R$. For any tuple
  $r \not\in eadom^{\DIN}_R$ we have $\RAeadomEV{R}(r) = 0$. From the
  properties of $\otimes$ we conclude that $F^{\DIN}(r) = 0$ for
  $r \not\in eadom^{\DIN}_R$.
  \item If $\mathcal{T}(\mathbf{r})$ is $\nabla\mathcal{T'}(\mathbf{r})$ or
  $\Delta\mathcal{T'}(\mathbf{r})$, then from the induction hypothesis we have
   a RA-expression $E$ on relation scheme $R$ corresponding to
  PTC-subexpression $\mathcal{T'}(\mathbf{r})$, such that
  $\mathcal{T'}^{\DIN}(r) = E^{\DIN}(r)$ for all $r \in \Tupl{R}$.

  We put $F = \nabla E$ or $F = \Delta E$, respectively.

  In both cases, the relation scheme of $F$ is $R$ as required. Assuming that
  the symbol $\Box$ denotes $\nabla$ or $\Delta$ we have
  \[
    \mathcal{T}^{\DIN}(r) = \mathcal{T}^{\DIN}( \| \mathbf{r} \|_r )
    = \Box \mathcal{T'}^{\DIN}( \| \mathbf{r} \|_r )
    = \Box \mathcal{T'}^{\DIN}( r )
    = \Box E^{\DIN}(r)
    = F^{\DIN}(r),
  \]
  for all $r \in eadom^{\DIN}_R$.

  For any tuple $r \not\in eadom^{\DIN}_R$ we have
  $\mathcal{T'}^{\DIN}( r ) = 0$ and from the induction hypothesis we also
  have $E^{\DIN}(r) = 0$. From the definition of $\nabla$ or $\Delta$ we
  conclude that $F^{\DIN}(r) = 0$ for $r \not\in eadom^{\DIN}_R$.

  \item If $\mathcal{T}(\mathbf{r})$ is $\ptcsupprime$ or $\ptcinfprime$
  where $\mathbf{r} = \mathbf{r_2}$, $R = R_{\mathbf{r}_2}$ and
  $R_{\mathbf{r}_1} \cap R_{\mathbf{r}_2} = \emptyset$, then from the
  induction hypothesis we have a RA-expression $E$ on relation scheme
  $R_{\mathbf{r}_1} \cup R$ corresponding to the PTC-subexpression
  $\mathcal{T'}(\mathbf{r}_1, \mathbf{r}_2)$, such that
  $\mathcal{T'}^{\DIN}(r r') = E^{\DIN}(r r')$ for all $r \in \Tupl{R}$ and
  $r' \in \Tupl{R_{\mathbf{r}_1}}$.
  Note that tuples $r$ and $r'$ are always joinable since the relation schemes
  $R_{\mathbf{r}_1}$ and $R$ are disjoint.

We distinguish two cases:
\begin{enumerate}
  \item if $\mathcal{T}(\mathbf{r})$ is $\ptcsupprime$, then we put
  $F = \pi_{R}(E)$.

  The relation scheme of $F$ is $R$ as required. We have
  \begin{align*}
    \mathcal{T}^{\DIN}(r) & =  \mathcal{T}^{\DIN}(\| \mathbf{r} \|_r) \\
      & = \textstyle\bigvee \{ \mathcal{T'}^{\DIN}( \| \mathbf{r_1}
        \cup
        \mathbf{r_2} \|_{r r'} )
        \mid
        r' \in eadom^{\DIN}_{R_{\mathbf{r}_1}} \} \\
      & = \textstyle\bigvee \{ \mathcal{T'}^{\DIN}(r r')
        \mid
        r' \in eadom^{\DIN}_{R_{\mathbf{r}_1}} \} \\
      & = \textstyle\bigvee \{ E^{\DIN}(r r')
        \mid
        r' \in eadom^{\DIN}_{R_{\mathbf{r}_1}} \} \\
      & = \textstyle\bigvee \{ E^{\DIN}(r r')
        \mid
        r' \in \Tupl{R_{\mathbf{r}_1}} \} \\
      & = \left(\pi_{R}(E)\right)^{\DIN}(r)
      = F^{\DIN}(r)
    \end{align*}
  for all $r \in eadom^{\DIN}_{R}$.

  Observe that extending the range of $r'$ from
  $eadom^{\DIN}_{R_{\mathbf{r}_1}}$ to $\Tupl{R_{\mathbf{r}_1}}$ cannot change
  the score $\bigvee\{ E^{\DIN}(r r')\}$ since for
  $r' \not\in eadom^{\DIN}_{R_{\mathbf{r}_1}}$ we have
  $\mathcal{T'}^{\DIN}(r r') = 0$ and thus $E^{\DIN}(r r') = 0$ for any
  $r \in \Tupl{R}$. Furthermore, for any $a \in L$ it holds that
  $a \vee 0 = a$. Hence, we have
  \[
    \textstyle\bigvee \{
      E^{\DIN}(r r')
      \mid
      r' \in eadom^{\DIN}_{R_{\mathbf{r}_1}}
    \}
    =
    \textstyle\bigvee \{ E^{\DIN}(r r')
    \mid
    r' \in \Tupl{R_{\mathbf{r}_1}} \}.
  \]

  Now we show that $F^{\DIN}(r) = 0$ for all $r \not\in eadom^{\DIN}_{R}$.
  For any $r \not\in eadom^{\DIN}_{R}$ we have $E^{\DIN}(r r') = 0$ and thus
  $F^{\DIN}(r) = \bigvee \{0, 0, \ldots \} = 0$.

  \item if $\mathcal{T}(\mathbf{r})$ is  $\ptcinfprime$, we put
  $F = E \div^{\RAeadom{R}} \RAeadom{R_{\mathbf{r}_1}}$.

  The relation scheme of $F$ is $R$ as required. We have
  \begin{align*}
    \mathcal{T}^{\DIN}(r) & =  \mathcal{T}^{\DIN}(\| \mathbf{r} \|_r) \\
    & = \textstyle\bigwedge \{
      \mathcal{T'}^{\DIN}( \| \mathbf{r_1} \cup \mathbf{r_2} \|_{r r'} )
      \mid
      r' \in eadom^{\DIN}_{R_{\mathbf{r}_1}} \} \\
    & = \textstyle\bigwedge \{ \mathcal{T'}^{\DIN}(r r')
      \mid
      r' \in eadom^{\DIN}_{R_{\mathbf{r}_1}} \} \\
    & = \textstyle\bigwedge \{ E^{\DIN}(r r')
      \mid
      r' \in eadom^{\DIN}_{R_{\mathbf{r}_1}} \} \\
    & = \textstyle\bigwedge \{ \RAeadomEV{R_{\mathbf{r}_1}}(r')
      \rightarrow
      E^{\DIN}(r r') \mid r' \in \Tupl{R_{\mathbf{r}_1}} \} \\
    & = \textstyle\bigwedge \{ \RAeadomEV{R}(r)
      \otimes
      \bigl(\RAeadomEV{R_{\mathbf{r}_1}}(r')
      \rightarrow
      E^{\DIN}(r r') \bigr) \mid r' \in \Tupl{R_{\mathbf{r}_1}} \} \\
    & = (E \div^{\RAeadom{R}} \RAeadom{R_{\mathbf{r}_1}})^{\DIN}(r)
      = F^{\DIN}(r)
    \end{align*}
  for all $r \in eadom^{\DIN}_{R}$.

  Note that extending the range of $r'$ from $eadom^{\DIN}_{R_{\mathbf{r}_1}}$
   to $\Tupl{R_{\mathbf{r}_1}}$ cannot change the final score of $r$, since
   for any $r' \not\in eadom^{\DIN}_{R_{\mathbf{r}_1}}$ we have
  $ \RAeadomEV{R_{\mathbf{r}_1}}(r') \rightarrow E^{\DIN}(r r') = 1 $ and it
  holds that $a \wedge 1 = a$ for any $a \in L$.

  For any tuple $r \not\in eadom^{\DIN}_{R}$ we have $\RAeadomEV{R}(r) = 0$.
  Hence, we have $F^{\DIN}(r) = \bigwedge \{0, 0, \ldots \} = 0$ for
  $r \not\in eadom^{\DIN}_R$.

  Observe that instead of using \eqref{def:rdiv} we can alternatively use
  Date's Small Divide and put
  $F' = \RAeadom{R} \div_{\text{gsdo}}^{E} \RAeadom{R_{\mathbf{r}_1}}$ since
  it holds that
  \begin{align*}
    \mathcal{T}^{\DIN}(r) & =  \mathcal{T}^{\DIN}(\| \mathbf{r} \|_r) \\
    & = \textstyle\bigwedge \{ \mathcal{T'}^{\DIN}( \| \mathbf{r_1}
      \cup
      \mathbf{r_2} \|_{r r'} )
      \mid
      r' \in eadom^{\DIN}_{R_{\mathbf{r}_1}} \} \\
    & = \textstyle\bigwedge \{ \mathcal{T'}^{\DIN}(r r')
      \mid
      r' \in eadom^{\DIN}_{R_{\mathbf{r}_1}} \} \\
    & = \textstyle\bigwedge \{ E^{\DIN}(r r')
      \mid
      r' \in eadom^{\DIN}_{R_{\mathbf{r}_1}} \} \\
    & = \textstyle\bigwedge \{ \RAeadomEV{R_{\mathbf{r}_1}}(r')
      \rightarrow
      E^{\DIN}(r r') \mid r' \in \Tupl{R_{\mathbf{r}_1}} \} \\
    & = \RAeadomEV{R}(r)
      \otimes
      \textstyle\bigwedge \{ \bigl(\RAeadomEV{R_{\mathbf{r}_1}}(r')
      \rightarrow
      E^{\DIN}(r r') \bigr) \mid r' \in \Tupl{R_{\mathbf{r}_1}} \} \\
    & = (\RAeadom{R} \div_{\text{gsdo}}^{E}
        \RAeadom{R_{\mathbf{r}_1}})^{\DIN}(r) = F'^{\DIN}(r)
    \end{align*}
  for all $r \in eadom^{\DIN}_{R}$.

  For any tuple $r \not\in eadom^{\DIN}_{R}$ we have $\RAeadomEV{R}(r) = 0$.
  From the properties of $\otimes$ we conclude that $F'^{\DIN}(r) = 0$ for any
  $r \not\in eadom^{\DIN}_{R}$.
\end{enumerate}
\end{enumerate}

\end{proof}

%%%%%%%%%%%%%%%%%%%%%%%%%%%%%%%%%%%%%%%%%%%%%%%%%%%%%%%%%%%%%%%%%%%%%%%%%%%%%%%%
%%%%%   FURTHER RELATIONSHIPS
%%%%%%%%%%%%%%%%%%%%%%%%%%%%%%%%%%%%%%%%%%%%%%%%%%%%%%%%%%%%%%%%%%%%%%%%%%%%%%%%
\section{More on Relationships of Division Operations}\label{sec:more}
In this section we use the Pseudo Tuple Calculus (PTC) to show further
relationships of the division operations presented in this paper.
We utilize the PTC in the following way. Let us have an relational operation 
$op$ that accepts input relations $\mathcal{D}_1, \ldots, \mathcal{D}_n$ on 
relation schemes $R_1, \ldots, R_n$ and its output relation is on relation 
scheme $R$.
For the input relations we consider relation symbols
$\blmathbb{D}_1, \ldots, \blmathbb{D}_n$ on the respective relation schemes
$R_1, \ldots, R_n$.
Note that the relation symbols are themselves RA-expressions. Now using the
relation symbols we construct a PTC-expression $\mathcal{T}(\mathbf{r})$ on $R$
that is semantically equivalent to the operation in question. By semantical
equivalence we mean that if we evaluate the PTC-expression
$\mathcal{T}(\mathbf{r})$ in a database instance $\DIN$ that maps the relation
symbols to the input relations, i.\,e. we have
$\blmathbb{D}_1^{\DIN} = \mathcal{D}_1,
 \ldots,
 \blmathbb{D}_n^{\DIN} = \mathcal{D}_n$, we get that
\[ op(\mathcal{D}_1, \ldots, \mathcal{D}_n)(r) = \mathcal{T}^{\DIN}(r) \]
for all $r \in \Tupl{R}$. Note that this construction does not depend on the
actual content of the input relations. Furthermore we apply the Theorem
\ref{PTC:PTCtoRA} to transform the PTC-expression $\mathcal{T}(\mathbf{r})$ to
an equivalent RA-expression that uses only the fundamental operations of the
algebra and obtain the requested relationship.

We give an example to illustrate the notion of semantical equivalence. Consider
the division operation defined by \eqref{def:rdiv}, i.\,e., for RDTs
$\mathcal{D}_1$, $\mathcal{D}_2$, and $\mathcal{D}_3$ on $RS$, $S$, and $R$,
respectively, the \emph{division}
$\rdiv{\mathcal{D}_3}{\mathcal{D}_1}{\mathcal{D}_2}$ of $\mathcal{D}_1$ by
$\mathcal{D}_2$ which ranges over $\mathcal{D}_3$ is an RDT on $R$ defined by
\begin{align*}
  \bigl(\rdiv{\mathcal{D}_3}{\mathcal{D}_1}{\mathcal{D}_2}\bigr)(r) =
  \textstyle\bigwedge_{s \in \mathrm{Tupl}(S)}
  \bigl(
  \mathcal{D}_3(r) \otimes (\mathcal{D}_2(s) \rightarrow  \mathcal{D}_1(rs))
  \bigr),
\end{align*}
for each $r \in \mathrm{Tupl}(R)$.
Consider relation symbols $\blmathbb{D}_1, \blmathbb{D}_2$ and
$\blmathbb{D}_3$ on $RS$, $S$, and $R$, respectively.
Then the PTC-expression
\begin{align*}
  \mathcal{T}(\blmathbb{r}) =
  \textstyle\bigwedge_{\blmathbb{s}}
  \bigl(
  \blmathbb{D}_3(\blmathbb{r}) \otimes (\blmathbb{D}_2(\blmathbb{s})
    \rightarrow
  \blmathbb{D}_1(\blmathbb{rs}))
  \bigr),
\end{align*}
is semantically equivalent to the division operation. More precisely, for a
database instance $\DIN$ such that $\blmathbb{D}_1^{\DIN} = \mathcal{D}_1$,
$\blmathbb{D}_2^{\DIN} = \mathcal{D}_2$, and
$\blmathbb{D}_3^{\DIN} = \mathcal{D}_3$ we have
\[ \bigl(\rdiv{\mathcal{D}_3}{\mathcal{D}_1}{\mathcal{D}_2}\bigr)(r)
  =
  \mathcal{T}^{\DIN}(r) \]
for all $r \in \Tupl{R}$.
Now, we are ready to show the relationships among the division operations.

\begin{theorem}\label{EQN:GSDObyOUR}
Let $\mathcal{D}_1$, $\mathcal{D}_2$, and $\mathcal{D}_3$ be RDTs on $RS$, $S$,
and $R$, respectively, and let $\div_\mathrm{gsdo}$ be Date's Small Divide.
For the division operation defined by \eqref{def:rdiv} we have
\[ \bigl(\rdiv{\mathcal{D}_3}{\mathcal{D}_1}{\mathcal{D}_2}\bigr)(r) =
(\RAeadomEV{R} \div_\mathrm{gsdo}^{E^{\DIN}} \RAeadomEV{S})(r), \]
for all $r \in \Tupl{R}$ where
\begin{align*}
E^{\DIN} & =  \mathcal{D}_3 \bowtie \bigl( (\mathcal{D}_2 \bowtie \RAeadomEV{R})
  \rightarrowtriangle^{\RAeadomEV{RS}}
  \mathcal{D}_1 \bigr)
\end{align*}
and the extended active domains $\RAeadomEV{R}, \RAeadomEV{S}$, and
$\RAeadomEV{RS}$ contain tuples built only from the values from relations
$\mathcal{D}_1, \mathcal{D}_2$, and $\mathcal{D}_3$.
\end{theorem}
\begin{proof}
First, using the relation symbols $\blmathbb{D}_1, \blmathbb{D}_2$, and
$\blmathbb{D}_3$, corresponding to the input relations $\mathcal{D}_1,
\mathcal{D}_2,$ and $\mathcal{D}_3$, we construct a PTC-expression
$\mathcal{T}(\blmathbb{r})$ that is semantically equivalent to the division
operation. We have
\[
  \mathcal{T}(\blmathbb{r}) =
    \textstyle\bigwedge_{\blmathbb{s}}
  \bigl(
  \blmathbb{D}_3(\blmathbb{r}) \otimes
  (\blmathbb{D}_2(\blmathbb{s}) \rightarrow  \blmathbb{D}_1(\blmathbb{rs}))
  \bigr)
\] and it holds that
$\bigl(\rdiv{\mathcal{D}_3}{\mathcal{D}_1}{\mathcal{D}_2}\bigr)(r) =
\mathcal{T}^{\DIN}(r)$ for all $r \in \Tupl{R}$ in any database instance $\DIN$
that maps the relation symbols to their respective input relations.
According to the Theorem \ref{PTC:PTCtoRA} there is an equivalent RA-expression
$F$ such that $\mathcal{T}^{\DIN}(r) = F^{\DIN}(r)$ for all $r \in \Tupl{R}$.
The sought RA-expression $F$ is
\[
F = \RAeadom{R} \div_\mathrm{gsdo}^{E} \RAeadom{S} \]
where
\[
E = \blmathbb{D}_3 \bowtie \bigl( ( \blmathbb{D}_2 \bowtie \RAeadom{RS} )
\rightarrowtriangle^{\RAeadom{RS}} ( \blmathbb{D}_1 \bowtie \RAeadom{S} ) \bigr).
\]
By evaluating $E$ in the database instance $\DIN$ that maps the relation symbols
 to their respective input relations we get a relation
\[
 E^{\DIN} = \mathcal{D}_3 \bowtie \bigl( ( \mathcal{D}_2 \bowtie \RAeadomEV{RS} )
 \rightarrowtriangle^{\RAeadomEV{RS}}
 ( \mathcal{D}_1 \bowtie \RAeadomEV{S} ) \bigr).
\]

The most simple database instance that maps the relation symbols
to their respective input relations contains just the relations
$\mathcal{D}_1, \mathcal{D}_2,$ and $\mathcal{D}_3$. The relations
$\RAeadomEV{RS}, \RAeadomEV{S},$ and $\RAeadomEV{R}$, obtained by evaluating
$\RAeadom{RS}, \RAeadom{S},$ and $\RAeadom{R}$, in such database instance
therefore contain tuples built only from the values from relations
$\mathcal{D}_1, \mathcal{D}_2$, and $\mathcal{D}_3$ as required.

It can be easily checked that even if the database instance contained more
relations and thus the extended active domains
$\RAeadomEV{RS}, \RAeadomEV{S},$ and $\RAeadomEV{R}$ contained more tuples built from
values from other relations these additional tuples do not change the result of
evaluating the RA-expression $F$. It is safe to build the
relations $\RAeadomEV{RS}, \RAeadomEV{S},$ and $\RAeadomEV{R}$, only from the
values from relations $\mathcal{D}_1, \mathcal{D}_2$, and $\mathcal{D}_3$.

From the properties of $\bowtie$ and the fact that $\RAeadomEV{S}$ contains the
projection of relation $\mathcal{D}_1$ to $S$ and the relation $\mathcal{D}_2$,
we can further simplify the form of the relation $E^{\DIN}$ to
\[
 E^{\DIN} = \mathcal{D}_3 \bowtie \bigl( ( \mathcal{D}_2 \bowtie \RAeadomEV{R} )
 \rightarrowtriangle^{\RAeadomEV{RS}} \mathcal{D}_1 \bigr).
\]
Putting all things together we have
\[
  \bigl(\rdiv{\mathcal{D}_3}{\mathcal{D}_1}{\mathcal{D}_2}\bigr)(r) =
  \mathcal{T}^{\DIN}(r) = F^{\DIN}(r) =
  (\RAeadomEV{R} \div_\mathrm{gsdo}^{E^{\DIN}} \RAeadomEV{S})(r)
\]
for all $r \in \Tupl{R}$ with the relations $E^{\DIN}$ and
$\RAeadomEV{R}, \RAeadomEV{S}, \RAeadomEV{RS}$ defined as above.
\end{proof}

\begin{theorem}
Let $\mathcal{D}_1$, $\mathcal{D}_2$, and $\mathcal{D}_3$ be RDTs on $R$, $S$,
and $RS$, respectively, and let $\div$ be the division operation defined by
\eqref{def:rdiv}.
For Date's Small Divide we have
\[ \bigl(
  \mathcal{D}_1 \div^{\mathcal{D}_3}_\mathrm{gsdo} \mathcal{D}_2
  \bigr)(r) =
  \bigl( \mathcal{D}_1 \bowtie
  (E^{\DIN} \div^{\RAeadomEV{R}} \RAeadomEV{S}) \bigr)(r), \]
for all $r \in \Tupl{R}$ where
\begin{align*}
E^{\DIN} & =  \bigl( (\mathcal{D}_2 \bowtie \RAeadomEV{R})
\rightarrowtriangle^{\RAeadomEV{RS}} \mathcal{D}_3 \bigr)
\end{align*}
and the extended active domains $\RAeadomEV{R}, \RAeadomEV{S}$, and
$\RAeadomEV{RS}$ contain tuples built only from the values from relations
$\mathcal{D}_1, \mathcal{D}_2$, and $\mathcal{D}_3$.
\end{theorem}
\begin{proof}
Use similar arguments as in the proof of Theorem \ref{EQN:GSDObyOUR}.
\end{proof}

\begin{theorem}
  Let $\mathcal{D}_1$, $\mathcal{D}_2$, $\mathcal{D}_3$, and $\mathcal{D}_4$
  be RDTs on $R_1, R_2, R_3$, and $R_4$, respectively, and let $\div$ be the
  division operation defined by \eqref{def:rdiv}.
For Darwen's Divide we have
\[
 \bigl(
  \mathcal{D}_1
  \div^{\mathcal{D}_3,\mathcal{D}_4}_\mathrm{gddo}
  \mathcal{D}_2
  \bigr)(r) =
  \bigl( (\mathcal{D}_1 \bowtie \mathcal{D}_2) \bowtie
  (E^{\DIN} \div^{\RAeadomEV{R_1'}} \RAeadomEV{R_2'}) \bigr)(r), \]
for all $r \in \Tupl{R_1 \cup R_2}$ where
\begin{align*}
E^{\DIN} & =  \bigl( (\mathcal{D}_4 \bowtie \RAeadomEV{R_3'})
\rightarrowtriangle^{ \RAeadomEV{R_4'} }
(\pi_{R_3'}(\mathcal{D}_3) \bowtie \RAeadomEV{R_4} ) \bigr), \\
R_1'     & = (R_4 \cap (R_1 \cup R_2)) \cup (R_1 \cap R_3), \\
R_2'   & = R_4 \setminus (R_1 \cup R_2), \\
R_3'   & = R_3 \cap (R_1 \cup R_4), \\
R_4'   & = R_4 \cup (R_1 \cap R_3)
\end{align*}
and the extended active domains contain tuples built only from the values from
relations $\mathcal{D}_1, \mathcal{D}_2, \mathcal{D}_3$, and $\mathcal{D}_4$.
\end{theorem}
\begin{proof}
As in the previous proofs, we construct PTC-expression $\mathcal{T}(\mathbf{r})$
 that is semantically equivalent to the division operation defined by
\eqref{def:gradedDarwenNoCondAlt}. Using the relation symbols $\blmathbb{D}_1,
\blmathbb{D}_2, \blmathbb{D}_3, $ and $\blmathbb{D}_4$ that correspond to the
input relations $\mathcal{D}_1, \mathcal{D}_2, \mathcal{D}_3$, and
$\mathcal{D}_4$ we get
\[
  \mathcal{T}(\mathbf{r}) =
  ( \blmathbb{D}_1 \bowtie \blmathbb{D}_2 ) (\mathbf{r}) \otimes
    \textstyle\bigwedge_{\mathbf{r'_{\text{b}}}}
  \bigl(
  \blmathbb{D}_4(\mathbf{r'_\text{f}}, \mathbf{r'_\text{b}})
  \rightarrow
  \pi_{R_3'}( \blmathbb{D}_3 )( \mathbf{r''_\text{f}, r''_\text{b} } )
  \bigr)
\]
where
\begin{itemize}
  \item $\mathbf{r}$ is on a relation scheme $R_1 \cup R_2$,
  \item $\mathbf{r'_{\text{b}}}$ is on $ R_2' = R_4 \setminus (R_1 \cup R_2)$,
  \item $\mathbf{r'_\text{f}}$ is on $ R_4 \cap (R_1 \cup R_2)$,
  \item $ R_3' = R_3 \cap (R_1 \cup R_4) $,
  \item $\mathbf{r''_\text{f}}$ is on $(R_1 \cup (R_2 \cap R_4)) \cap R_3'$,
  \item $\mathbf{r''_\text{b}}$ is on $R_2' \cap R_3'$
\end{itemize}
such that each set of tuple variables contains one tuple variable for each
attribute in the relation schema of the corresponding subexpression.
For instance, the set of tuple variables $\mathbf{r}$ can be characterized as
$\mathbf{r} = \{ \blmathbb{r}_y \mid y \in R_1 \cup R_2 \}$.

According to the Theorem \ref{PTC:PTCtoRA} there is an equivalent RA-expression
$F$ such that $\mathcal{T}^{\DIN}(r) = F^{\DIN}(r)$ for all $r \in \Tupl{R}$.
Again, the database instance $\DIN$ should map each relation symbol to its
corresponding input relation.
In order to find the RA-expression $F$, we first find the RA-expression $E$ that
 corresponds to the PTC-subexpression
$\blmathbb{D}_4(\mathbf{r'_\text{f}}, \mathbf{r'_\text{b}})
\rightarrow
\pi_{R_3'}( \blmathbb{D}_3 )( \mathbf{r''_\text{f}, r''_\text{b} } )$.
The sought RA-expression is
\[
  E = \left( \left(\blmathbb{D}_4 \bowtie \RAeadom{R_3'} \right)
  \rightarrowtriangle^{ \RAeadom{ R_4' }}
  \left(\pi_{R_3'}(\blmathbb{D}_3) \bowtie \RAeadom{R_4} \right) \right)
\]
where $R_4' = R_4 \cup R_3' = R_4 \cup (R_3 \cap (R_1 \cup R_4)) =
R_4 \cup (R_1 \cap R_3) $.

Now, we are ready to find the RA-expression $F$ that corresponds to the whole
PTC-expression $\mathcal{T}(\mathbf{r})$. We have
\[
F = ( \blmathbb{D}_1 \bowtie \blmathbb{D}_2 )
  \bowtie
    ( E \div^{\RAeadom{R_1'}} \RAeadom{R_2'}) \]
where $R_1' = R_4' \setminus R_2' =
(R_4 \cap (R_1 \cup R_2)) \cup (R_1 \cap R_3). $
Since it holds that $R_1' \subseteq (R_1 \cup R_2)$ the relation scheme of $F$
is $R_1 \cup R_2$ as required.

The rest of the proof is clear.
\end{proof}

In the previous chapters, we have already shown that Date's Small Divide is a
special case of Date's Great Divide which is in turn a special case of Darwen's
Divide. Furthermore, we have shown that if the $\mathbf{L}$ is prelinear or
divisible, then there is a simple correspondence between Date's Small Divide and
 the division operation defined by \eqref{def:rdiv}.

In this chapter we have shown that they are equivalent regardless of the
properties of $\mathbf{L}$. We have also shown that Darwen's Divide can be
expressed by the division operation defined by~\eqref{def:rdiv}. As a consequence
 we get the equivalence of all domain-independent division operations presented
 in this paper. Furthermore we have an exact way to express one division using
 the other. Therefore we can summarize the observations as follows:

\begin{corollary}
All domain-independent division operations presented in this paper are
equivalent. \qed
\end{corollary}
This result solves an open question concerning the
relationship of Date's Great Divide and Darwen's Divide in the classic
setting, see \cite[page 187]{DaDe:DErd}.

%%%%%%%%%%%%%%%%%%%%%%%%%%%%%%%%%%%%%%%%%%%%%%%%%%%%%%%%%%%%%%%%%%%%%%%%%%%%%%%%
%%%%%   CONCLUSION
%%%%%%%%%%%%%%%%%%%%%%%%%%%%%%%%%%%%%%%%%%%%%%%%%%%%%%%%%%%%%%%%%%%%%%%%%%%%%%%%
\section{Conclusion}
We have presented a survey of graded generalizations of classic division-like
operations in a rank-aware model of data. We have focused on generalizing
variants of division-like operations which are neglected by other
rank-aware approaches in databases.
In our model we assume that~\eqref{def:rdiv} is a fundamental
operation. Under this assumption,
we have shown that all the graded generalizations of the classic division
operations we have studied in this paper are derived operations.
That is, considering the original graded division~\eqref{def:rdiv} as
the fundamental division, i.e., including it in the relational algebra,
all the other divisions~\eqref{def:gSmallOriginal}, \eqref{def:gSmall},
\eqref{def:gGreatOriginal} and \eqref{def:gradedDarwenNoCondAlt}, are derived 
operations in our model. Furthermore,
using the Pseudo Tuple Calculus (PTC), we have shown that the various
variants of the division operations are mutually definable. Interestingly,
some of our observations we have made on the general level
(considering $\mathbf{L}$ as a general complete residuated lattice)
pertain to the classic model---when $\mathbf{L}$ is considered as
the two-element Boolean algebra. For instance, we have shown that
Date's Great Divide and Darwen's Divide are mutually definable. This result
solves an open question that was stated by Date in \cite[page 187]{DaDa:3rdM}.

Future research in the area may include considerations on the role of
fundamental and derived operations in the model. The fundamental
division~\eqref{def:rdiv} cannot be dropped without losing the
expressive power of the relational algebra since in general we
cannot introduce universal quantifiers using the existential ones.
On the other hand, there may be ways to simplify the present relational
algebra by considering other forms of division-like operations.
One way to go is to introduce graded subsethood as a fundamental
(graded) comparator of relations, and use analogous techniques as
image relations~\cite{DaDe:DEim} to express the division.

%%%%%%%%%%%%%%%%%%%%%%%%%%%%%%%%%%%%%%%%%%%%%%%%%%%%%%%%%%%%%%%%%%%%%%%%%%%%%%%%
\subsection*{Acknowledgment}
Supported by grant no. \verb|P202/14-11585S| of the Czech Science Foundation.
\newline
O.~Vaverka was also supported by internal student grant \verb|IGA_PrF_2015_023|
of Palacky University Olomouc.

%%%%%%%%%%%%%%%%%%%%%%%%%%%%%%%%%%%%%%%%%%%%%%%%%%%%%%%%%%%%%%%%%%%%%%%%%%%%%%%%
%%%%%   REFERENCES
%%%%%%%%%%%%%%%%%%%%%%%%%%%%%%%%%%%%%%%%%%%%%%%%%%%%%%%%%%%%%%%%%%%%%%%%%%%%%%%%

\bibliographystyle{amsplain}
\bibliography{ovvv}

% that's all folks
\end{document}